\documentclass[journal]{IEEEtran}

\usepackage{graphicx,amssymb,lineno}
\usepackage{amsmath,amsfonts,amssymb}
\usepackage{amsthm}
\newtheorem{theorem}{Theorem}
\usepackage{algorithm}
\usepackage{algorithmic}
\usepackage{setspace}
\usepackage[usenames]{color}
\usepackage{float}
\usepackage{mathtools}
\usepackage{cite}
\usepackage{epstopdf}
\usepackage[numbers,sort&compress]{natbib}
\usepackage{array}
\allowdisplaybreaks[4]

\renewcommand{\algorithmicrequire}{\textbf{Initialization:}}

\usepackage{graphics,color,epsfig,graphpap,rotate}
\usepackage{times,verbatim,epsfig,graphicx,latexsym}
\usepackage{url}
\usepackage{subfigure}
\usepackage{booktabs}

\begin{document}

\title{Sum Rate Maximization under AoI Constraints for RIS-Assisted mmWave Communications}

\author{Ziqi~Guo,
        Yong~Niu,~\IEEEmembership{Senior Member,~IEEE,}
        Shiwen~Mao,~\IEEEmembership{Fellow,~IEEE,}
        Changming~Zhang,
        Ning~Wang,~\IEEEmembership{Member,~IEEE,}
        Zhangdui~Zhong,~\IEEEmembership{Fellow,~IEEE,}
        and~Bo~Ai,~\IEEEmembership{Fellow,~IEEE}
\thanks{Copyright (c) 2015 IEEE. Personal use of this material is permitted. However, permission to use this material for any other purposes must be obtained from the IEEE by sending a request to pubs-permissions@ieee.org. This work was supported in part by the National Key Research and Development Program of China under Grant 2021YFB2900301, in part by the National Key Research and Development Program of China under Grant 2020YFB1806903, in part by the National Natural Science Foundation of China under Grants 62221001, 62231009, and U21A20445, in part by the Fundamental Research Funds for the Central Universities, China, under Grants 2022JBQY004 and 2022JBXT001, and in part by the Fundamental Research Funds for the Central Universities under Grant 2023JBMC030. (\emph{Corresponding authors: Yong Niu; Bo Ai.})}
\thanks{Ziqi Guo is with the State Key Laboratory of Advanced Rail Autonomous Operation, Beijing Jiaotong University, Beijing 100044, China, and also with the Collaborative Innovation Center of Railway Traffic Safety, Beijing Jiaotong University, Beijing 100044, China (e-mail: 21120053@bjtu.edu.cn).}
\thanks{Yong Niu is with the State Key Laboratory of Advanced Rail Autonomous Operation, Beijing Jiaotong University, Beijing 100044, China (e-mail: niuy11@163.com).}
\thanks{Shiwen Mao is with the Department
of Electrical and Computer Engineering, Auburn University,
Auburn, AL, 36849-5201 USA (e-mail: smao@ieee.org).}
\thanks{Changming Zhang is with the Research Institute of Intelligent Networks, Zhejiang Lab, Hangzhou 311121, China (e-mail: zhangcm@zhejianglab.com).}
\thanks{Ning Wang is with the School of Information Engineering, Zhengzhou University, Zhengzhou 450001, China (e-mail: ienwang@zzu.edu.cn).}
\thanks{Zhangdui Zhong, and Bo Ai are with the State Key Laboratory of Advanced Rail Autonomous Operation, Beijing Jiaotong University, Beijing 100044, China, and also with the Beijing Engineering Research Center of High-speed Railway Broadband Mobile Communications, Beijing Jiaotong University, Beijing 100044, China (e-mails: zhdzhong@bjtu.edu.cn; aibo@ieee.org).}}
\maketitle

\begin{abstract}
The concept of age of information (AoI) has been proposed to quantify information freshness, which is crucial for time-sensitive applications. However, in millimeter wave (mmWave) communication systems, the link blockage caused by obstacles and the severe
path loss greatly impair the freshness of information received by the user equipments (UEs). In this paper, we focus on reconfigurable intelligent surface (RIS)-assisted mmWave communications, where beamforming is performed at transceivers to provide directional beam gain and a RIS is deployed to combat link blockage.
We aim to maximize the system sum rate while satisfying the information freshness requirements of UEs by jointly optimizing the beamforming at transceivers, the discrete RIS reflection coefficients, and the UE scheduling strategy.
To facilitate a practical solution, we decompose the problem into two subproblems. For the first per-UE data rate maximization problem, we further decompose it into a beamforming optimization subproblem and a RIS reflection coefficient optimization subproblem. Considering the difficulty of channel estimation, we utilize the hierarchical search method for the former and the local search method for the latter, and then adopt the block coordinate descent (BCD) method to alternately solve them. For the second scheduling strategy design problem, a low-complexity heuristic scheduling algorithm is designed. Simulation results show that the proposed algorithm can effectively improve the system sum rate while satisfying the information freshness requirements of all UEs.
\end{abstract}


\section{Introduction}\label{S1}

\IEEEPARstart{I}{n} recent years, a variety of novel applications have emerged, leading to a dramatic increase in mobile data traffic.
According to the International Telecommunication Union (ITU), mobile data traffic is predicted to grow from 62 EB per month in 2020 to 5$,$016 EB per month in 2030~\cite{01}.
This puts a compelling need for higher capacity of communication systems and further intensifies the conflict between the demands for communication capacity and the scarce spectrum resources.
Therefore, millimeter wave (mmWave) communication is
considered as a promising technology for future cellular
networks due to its large available bandwidth~\cite{02,02b,02c}.

On the other hand, many new time-sensitive applications, e.g., autonomous driving, depend on timely and reliable information exchange. Once information is generated, it should be sent to the receiver for timely processing, and outdated information could seriously degrade the user's experience.
In order to capture the information freshness, age of information (AoI) has been proposed, which is defined as the elapsed time since the generation of the most recently received status-update~\cite{08,08b}.
Receivers wish to receive data with a lower AoI, so that the received data will be fresher. The data received with a high AoI could be meaningless or harmful.

In practical communication systems, AoI is generally affected by the scheduling strategy and the quality of the received signal. However, compared with microwave communication below 6 GHz, a key challenge of mmWave communication is that the signal in the mmWave band will experience more severe path loss due to the short wavelength~\cite{03}, which degrades the quality of the received signal. It is necessary to establish a directional transmission link between transceivers with the help of large-scale antenna arrays and beamforming, which can provide high antenna gains for mmWave signals to compensate for path loss. However, the directional transmissions and weak diffraction ability make mmWave signals vulnerable to blockage, especially in indoor and dense urban environments~\cite{04}. The high AoI due to the blockage nature is often unacceptable in most time-sensitive applications.

Fortunately, reconfigurable intelligent surface (RIS) can flexibly configure the propagation environment through software programming, which can be used to combat mmWave link blockage~\cite{05}. Specifically, RIS is a device composed of a large number of passive reconfigurable reflection elements. Each element can independently control the amplitude and phase changes to the incident signal in a software-defined manner~\cite{06}. By a proper design, the passive reflections of all the reflection elements of RIS can be coherently superposed at the desired receiver to increase the received signal power, thus creating a more reliable reflection link and avoiding blockage of the signal in the direct link. Therefore, deploying RIS in mmWave systems exhibits the potential to achieve superior information freshness performance. However, as a passive reflective device, RIS is not capable of transmitting, processing, and receiving signals. The task of channel estimation grows increasingly challenging as the count of reflection elements escalates. Besides, compared with other wireless systems, the large-scale antenna arrays in mmWave systems greatly increase the difficulty of channel estimation~\cite{09}. Therefore, it is necessary to discuss how to guarantee the information freshness performance of the system without knowing CSI.

Therefore, in this paper, we study a downlink RIS-assisted mmWave MIMO system, where the base station (BS) transmits time-sensitive data to user equipments (UEs). Considering the difficulty of channel estimation in the system, we assume that the full CSI is unknown. Different from most of the existing AoI research on the overall AoI minimization, we wish to satisfy the information freshness requirement of each UE in the system, which can provide a better communication experience for UEs. Besides, we pursue the maximization of the system sum rate while satisfying the information freshness requirements, which can further stimulate the potential of mmWave communication systems. Thus, our work aims to maximize the system sum rate over a fixed time interval, i.e., a superframe, while satisfying all the UEs' information freshness requirements in the system. Optimization variables include the beamforming vectors at the BS and UEs, the discrete RIS reflection coefficients, and the scheduling matrix, all of which are coupled in the expressions of the sum rate in the objective function and AoI constraints. The problem is an integer non-convex optimization problem. To reduce the complexity of the solution, we decompose the optimization problem into several subproblems and solve them separately.

The contributions of this paper are summarized as follows:
\begin{itemize}
  \item
  We study a downlink RIS-assisted mmWave MIMO system, in which a RIS is deployed to provide a reliable reflection path against the blockage of direct links, and time-sensitive data is transmitted from the BS to UEs. Each UE in the system has certain requirement for the freshness of information.

  \item
    We formulate the system sum rate maximization problem by optimizing the beamforming vectors at the BS and UEs, the RIS reflection coefficients, and the scheduling matrix, subject to the AoI constraints of UEs. Since all the optimization variables are coupled, the complexity of finding the overall optimal solution by exhaustive search will be prohibitively high. To address this issue, we decompose the original problem into a per-UE rate maximization problem and a scheduling strategy design problem.

  \item
  For the per-UE rate maximization problem, considering that the full CSI is unknown, the hierarchical search method and local search method are used for the optimization of the beamforming vectors and the RIS reflection coefficients, respectively.
  Due to the coupled beamforming vectors and RIS reflection coefficients,
  we use the block coordinate descent (BCD) algorithm to iteratively update the two sets of optimization variables.
  For the scheduling strategy design problem, we propose a low-complexity heuristic strategy, which maximizes the system sum rate over a superframe while satisfying  the AoI constraints.

  \item
  We evaluate the performance of the proposed algorithm with simulations.
  Compared with three benchmark schemes, the simulation results demonstrate that the proposed algorithm ensures the information freshness requirements of all UEs, and the system sum rate is effectively improved.

\end{itemize}

The rest of the paper is organized as follows. Section~\ref{S2}
reviews related work.
The system overview
and problem formulation are presented
in Section~\ref{S3} and Section~\ref{S4}, respectively.
We present the sum rate maximization algorithm
in Section~\ref{S5} and discuss our simulation results
in Section~\ref{S6}.
Section~\ref{S7} concludes this paper.
%
%
%
%
%
%
%

\section{Related Work}\label{S2}
In RIS-assisted communication
systems, a key problem of interest is to jointly devise the RIS reflection coefficients and the active beamforming vectors at the BS to improve system performance.
Numerous studies have been
conducted to solve this problem under different system
setups and assumptions, some of which are for mmWave MIMO communication systems.
Perovic \emph{et al.}~\cite{10} compared two optimization schemes in the indoor RIS-assisted mmWave environment without the line-of-sight (LOS) path. They showed the joint optimization of the RIS reflection elements and the transmit phase precoder can effectively enhance channel capacity.
Wang \emph{et al.}~\cite{11} considered a RIS-assisted downlink mmWave
system with a hybrid beamforming structure. A manifold optimization (MO)-based algorithm was developed to jointly optimize the RIS's reflection
coefficients and the hybrid beamforming at the BS for maximization
of spectral efficiency.
Feng \emph{et al.}~\cite{12} designed a successive
interference cancelation (SIC)-based method for the bandwidth-efficiency maximization problem.
A greedy method is proposed for the hybrid beamforming design
and a complex circle manifold (CCM)-based method is used for updating of the RIS elements.
Li \emph{et al.}~\cite{13} formulated a power minimization problem with signal-to-interference-plus-noise ratio (SINR) constraints in multi-user scenarios. They proposed a two-layer penalty-based algorithm
to decouple variables in SINR constraints and
three different
methods to optimize the BS analog beamforming
and the RIS response matrix in the penalty-based algorithm.

The above works rely on full CSI through channel estimation.
Considering the difficulty of channel estimation in RIS-assisted mmWave MIMO systems,
some studies focus on the problem of beam training with the goal of obtaining the angle of departure (AoD) and angle of arrival (AoA) associated with the dominant path.
Wang  \emph{et al.}~\cite{14} developed an efficient
downlink beam training method for RIS-assisted
mmWave or THz systems. They designed multi-directional beam training sequences
to scan the angular space and proposed an efficient set-intersection-based scheme to identify the best beam alignment.
Wei \emph{et al.}~\cite{15} proposed an effective near-field beam training scheme by designing a near-field codebook that matches the near-field channel model for the extremely
large-scale RIS-assisted system.
Wang \emph{et al.}~\cite{16} considered a multi-RIS-assisted mmWave MIMO system and carried out beam training designs with random beamforming and maximum likelihood (ML) estimation to estimate the parameters of the LOS component.


Recently, the importance of information freshness has been recognized and AoI has been considered in the design of wireless communication systems.
For example, He \emph{et al.}~\cite{17} designed two scheduling strategies for the system AoI minimization problem in wireless networks. One is based on the ILP formulation, and the other is the suboptimal but more scalable steepest age descent algorithm.
Kadota \emph{et al.}~\cite{18} formulated the problem of minimizing the expected weighted sum of AoI under time-throughput constraints. They designed four low-complexity scheduling strategies for solving this problem and found the Max-Weight and the Drift-Plus-Penalty have better performance in terms of AoI and throughput.
Liu \emph{et al.}~\cite{19} proved that any optimal solution of the maximum delay minimization problem is an approximate solution of the AoI minimization problem with bounded optimality loss. Inspired by this, a framework was developed to solve the AoI minimization problem in multi-path communications.
Bhat \emph{et al.}~\cite{20} considered the long-term average throughput maximization problem in fading channels, where the system average AoI and power are regarded as constraints. They proposed a simple age-independent stationary randomized power allocation policy to solve the problem.
In addition, the optimization of AoI has been extended to different application scenarios, such as multi-access edge computing-assisted IoT networks~\cite{21}, unmanned aerial vehicles (UAV) communications~\cite{22}, simultaneous wireless information and power transfer (SWIPT) enabled communications~\cite{23}, and the joint radar-communication (JRC)~\cite{24}, etc.

There were also some related works on AoI optimization in RIS-assisted wireless communications~\cite{25,26,27,28,29,30,31}.
Sorkhoh \emph{et al.}~\cite{25} studied the RIS-assisted cooperative autonomous driving (CAD) systems. They scheduled the resource blocks and RISs to minimize the average AoI of all streams.
Muhammad \emph{et al.}~\cite{26} examined the joint optimization of the RIS phase shifts and the traffic streams scheduling based on semi-definite relaxation (SDR), and solved the problem of minimizing the sum AoI in RIS-assisted wireless networks in single-antenna scenarios.
Samir \emph{et al.}~\cite{27} formulated an optimization problem with the objective of minimizing the expected sum AoI in an IoT network with the relay of a UAV equipped with RIS. To solve this problem, they developed a deep reinforcement learning (DRL) framework to jointly optimize the UAV height, RIS phase shift, and scheduling strategy.
Fan \emph{et al.}~\cite{28} deployed a RIS between IoT devices and UAVs to overcome the obstacles of urban buildings, and designs a DRL scheme to optimize UAV trajectory, discrete RIS phase shift, and scheduling strategy to minimize the total AoI of all devices.
Feng \emph{et al.}~\cite{29} adopted the DRL algorithm to jointly optimize the phase-shift matrix of RIS and service time of packets to solve the problem of minimizing the average peak information age in RIS-assisted non-orthogonal multiple access (NOMA) networks.
Lyu \emph{et al.}~\cite{30} investigated the sum AoI minimization in a RIS-assisted SWIPT network, where the energy harvesting demands of users were considered.
They proposed a successive convex approximation (SCA) based alternating optimization (AO) algorithm to handle the scheduling problem with joint active and passive beamforming design.
Shi \emph{et al.}~\cite{31} considered the average AoI minimization through the joint design of the RIS phase shifts,
transmit powers, and transmission rate in hybrid automatic repeat request (HARQ)-RIS aided
IoT networks.
However, all of these studies took the overall AoI minimization as their objective and neglected the information freshness requirements of UEs.
Moreover, the AoI optimization in RIS-assisted mmWave communications has not been studied yet.
Therefore, in this paper, we focus on sum rate maximization while satisfying the information freshness requirements of UEs in RIS-assisted mmWave MIMO systems.

\section{System Overview}\label{S3}

\subsection{System Model}\label{S3-1}
As shown in Fig.~\ref{fig:1}, we consider a single-cell mmWave MIMO communication system, where multiple UEs need to obtain fresh data from the BS.
A typical example is that UEs require real-time traffic information from the BS for trip planning.
The set of UEs is denoted as $\mathcal{K} = \{1, 2, \cdots, K\}$.
The BS is equipped with $N_t$ antennas and each UE is equipped with $N_r$ antennas.
Both the BS and the UEs use the uniform linear array (ULA) antennas.
The direct links between the BS and the UEs are assumed to be blocked by some obstacles, e.g., high buildings.
Thus, a RIS with $M$ passive reflection elements is deployed to provide a reliable reflection link for the UEs.

Moreover, time is divided into a series of non-overlapping superframes. Each superframe consists of two phases: the scheduling phase and the transmission phase. In the scheduling phase,
the scheduling and network optimization scheme is computed by a central controller located at the BS, then the results are sent to the RIS and the UEs.
In the transmission phase,
the BS communicates with the UEs with the help of RIS following the scheme.
If the system changes during the transmission phase causing a transmission failure, UEs will report the transmission failure to the BS.
The system will advance to the next superframe and the BS will redesign the scheme.

\begin{figure}[htbp]
\begin{center}
\includegraphics*[width=1.0\columnwidth,height=2.0in]{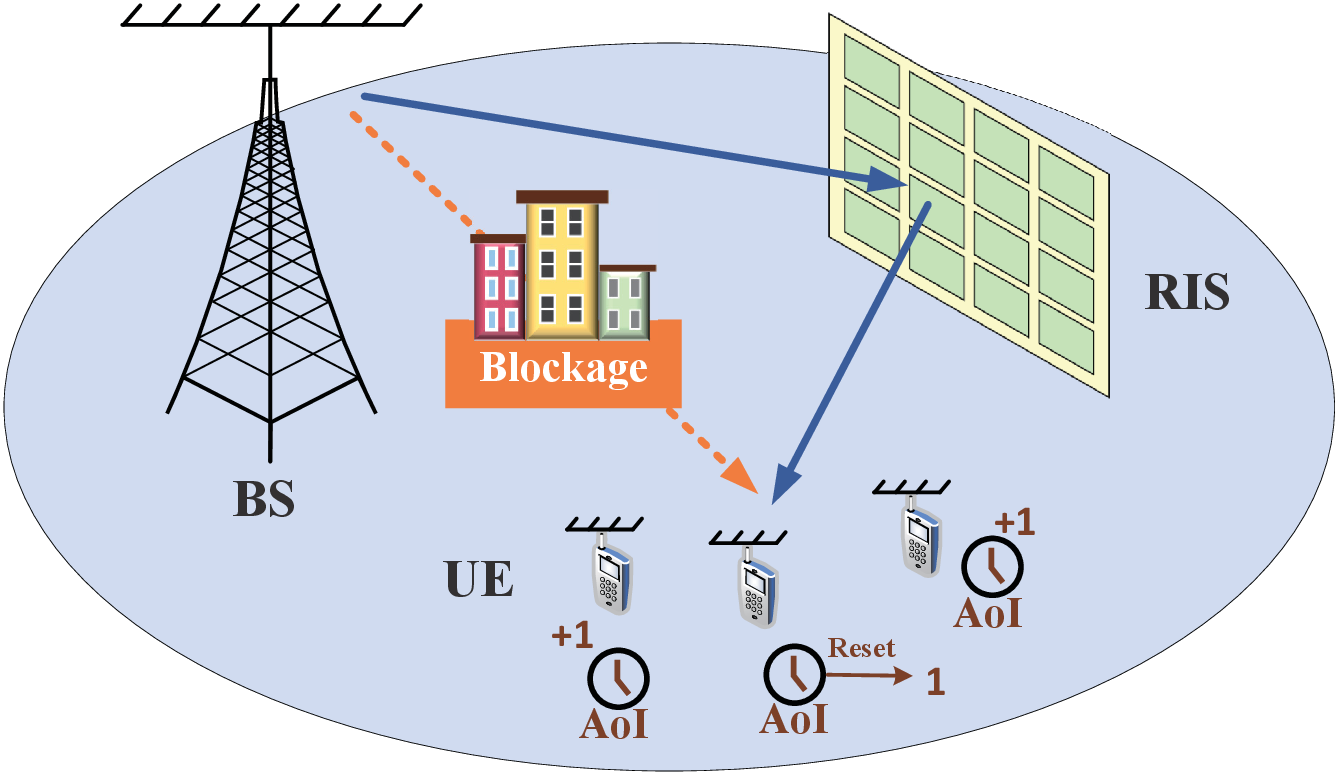}
\end{center}
\caption{An illustration of the system model.}
\label{fig:1}
\end{figure}

We focus on the performance of the system in the transmission phase. The transmission phase can be divided equally into $T$ time slots, denoted as $\mathcal{T} = \{1, 2, \cdots, T\}$. The time-division multiple access (TDMA) protocol is adopted, which means that only one UE can be scheduled for each time slot, so there is no interference between different UEs~\cite{28,31b}.
At each time slot, the BS side and the scheduled UE side perform analog beamforming to generate directional antenna gain. We assume that for each time slot, the BS transmits a signal with the same power $P_{T}$. The signal transmitted by the BS in time slot $t$ can be expressed as $\textbf{x}_t = \textbf{w}_t \sqrt{P_T} s_{t}$, where $s_{t}$ denotes the transmitted data at time slot $t$ with $\mathbb{E}\{s_{t}\}=0$ and $\mathbb{E}\{s_{t}s_{t}^H\}=1$, and $\textbf{w}_t \in \mathbb{C}^{N_t \times 1} $ denotes corresponding beamforming vector at the BS. The received signal at the scheduled UE in time slot $t$ is expressed as
\begin{equation}
y_t={\mathbf{f}_t^H} ({\mathbf{H}_t^H}\mathbf{x}_t+\mathbf{n}_t)={\mathbf{f}_t^H} ({\mathbf{H}_t^H}\mathbf{w}_t\sqrt{P_T}s_t+\mathbf{n}_t),
\label{eq1}
\end{equation}
where $\mathbf{f}_t \in \mathbb{C}^{N_r \times 1}$ denotes the beamforming vector at the UE, $\mathbf{n}_t \sim \mathcal{CN}(\mathbf{0},\sigma^2\mathbf{I}_{N_r})$ is the additive Gaussian white noise received by the UE,
and $\mathbf{H}_t \in\mathbb{C}^{N_t\times N_r}$ denotes the channel matrix in time slot $t$.

Since the direct link is blocked, the transmitted signal arrives at the UE via the BS-RIS-UE channel.
The RIS is a uniform planar array (UPA) consisting of $M$ passive reflection elements, each of which can independently adjust the amplitude and phase of the incident signal.
In view of the severe path loss, we ignore the signals reflected by the RIS twice and more and consider only the signal reflected for the first time~\cite{32}.
Let $\mathbf{G}_t\in\mathbb{C}^{N_t\times M}$ and $\mathbf{H}_{r,t} \in \mathbb{C}^{M\times N_r}$ represent the reflection channel matrixes at time slot $t$ from the BS to the RIS and from the RIS to the UE, respectively.
Thus, the channel matrix $\mathbf{H}_t$ can be expressed as
\begin{equation}
\mathbf{H}_t=\mathbf{G}_t\mathbf{\Phi}_t\mathbf{H}_{r,t}.
\label{eq2}
\end{equation}
Here, $\mathbf{\Phi}_t = \text{diag} ({\beta_{1,t}e^{j\varphi_{1,t}},\ \ \beta_{2,t}e^{j\varphi_{2,t}},\ \cdots\ ,\ \beta_{M,t}e^{j\varphi_{M,t}}}) \in \mathbb{C}^{M\times M}$ denotes the reflection-coefficient matrix of the RIS, where $\beta_{m,t} \in [0,1]$ and $\varphi_{m,t} \in [0,2\pi]$ represent the amplitude reflection coefficient and phase-shift reflection coefficient of RIS element $m$ in time slot $t$, respectively.
For simplicity, each reflection element of RIS is designed to maximize signal reflection (i.e., $\beta_{m,t}=1,\ \forall m,t)$~\cite{33,33b}.
Further, for the sake of hardware implementation, the phase shift of RIS takes finite discrete values.
We assume that each RIS element can realize $2^b$ different discrete phase shift values by $b$-bit quantization, the set of discrete phase shifts is represented as
$\mathcal{F}=\left\{0,\frac{2\pi}{2^{b}},\cdots,\left(2^b-1\right)\frac{2\pi}{2^{b}-1}\right\}$~\cite{07}.

Accordingly, the SNR received by the UE scheduled in time slot $t$ is given by
\begin{equation}
\gamma_t=\frac{\left|{\mathbf{f}_t^H}{(\mathbf{G}_t\mathbf{\Phi}_t\mathbf{H}_{r,t})}^H\mathbf{w}_t\sqrt{P_T}\right|^2}{\sigma^2}.
\label{eq3}
\end{equation}
To ensure that the UE can correctly demodulate the signal, the SNR should be greater than a threshold value $\gamma_{th}$, i.e., $\gamma_t>\gamma_{th}$.
Then, the achievable transmission rate in time slot $t$ can be written as
\begin{equation}
R_t=\log_2{\left(1+\gamma_t\right)}.
\label{eq4}
\end{equation}

However, it is worth noting that obtaining the full CSI by channel estimation is difficult in this system~\cite{09}.
Moreover, the RIS in the system is a passive device with no RF chain. Therefore, it cannot receive, transmit, and process signals other than just reflecting signals. It cannot directly estimate the BS-RIS channel and the RIS-UE channel. On the other hand, the large antenna array and the large number of passive reflection elements of RIS impose a substantial overhead on channel estimation. This is very detrimental to the design and optimization of the system. Therefore, we adopt the beam training method in the scheduling phase to obtain the AoD and AoA associated with
the reflection path, instead of explicitly estimating
the entire channel.
Specially,
the beam search space is represented by
a codebook containing multiple codewords.
We denote the codebook of the BS and each UE as $\Gamma_t$ and $\Gamma_r$, respectively.
Thus, we have ${{{\bf{w}}_t}} \in \Gamma_t$ and ${{{\bf{f}}_t}} \in \Gamma_r$ for $\forall t$.
In the scheduling phase, the BS consecutively sends beam training signals to each UE through the reflection of the RIS.
Both the BS and each UE can sweep the beamforming vectors in the pre-designed codebook, while the different phase shift of each RIS element is selected from $\mathcal{F}$ to change the reflection beam direction.
Then,
based on the UE's feedback, the combination of beamforming vectors and RIS reflection coefficients that maximizes the UE's achievable transmission rate will be selected.

\subsection{Channel Model}\label{S3-3}
\begin{figure}[t]
\begin{center}
\includegraphics*[width=1.0\columnwidth,height=2.1in]{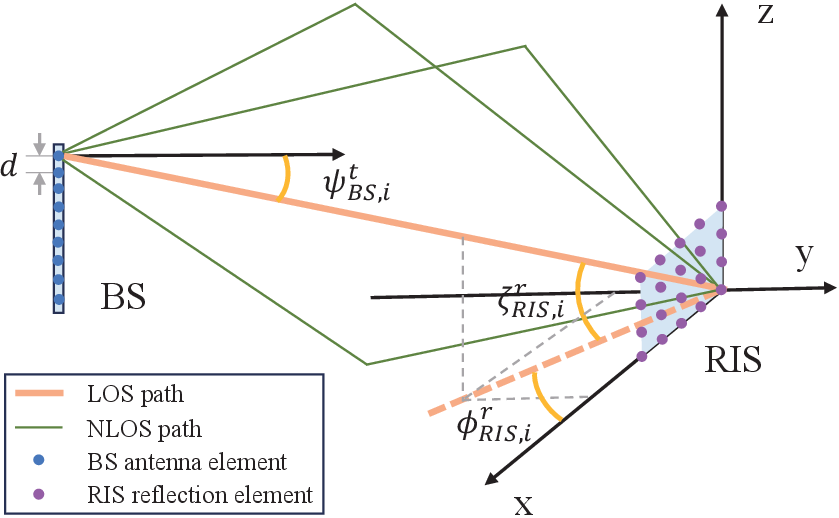}
\end{center}
\caption{An illustration of the BS-RIS channel model.}
\label{fig:channel}
\end{figure}
Due to the small wavelength, mmWave signals exhibit weak diffraction capabilities and suffer from high path loss, which makes the mmWave channel have limited scattering. The number of scatterers is typically substantially fewer than the number of antennas at the transceiver. Moreover, the dense configurations of antenna arrays in mmWave transceivers introduce pronounced antenna correlation. Given this, the Saleh-Valenzuela (S-V) channel model has been extensively used the capture the mathematical attributes of the mmWave channels~\cite{35}. In Fig.~\ref{fig:channel}, we present a schematic diagram of the BS-RIS channel based on the S-V channel model. The channel matrix between transceivers can be portrayed as a superposition of multipath components, where different multipath components have different angles of separation (AoDs) and angles of arrival (AoAs).

Assume that the channels do not change within a superframe.
In each time slot $t$, the BS-RIS channel $\mathbf{G}_t$ and the RIS-UE channel $\mathbf{H}_{r,t}$ can be written as
\begin{equation}
\mathbf{G}_t=\sqrt{\frac{N_tM}{P}}\mathop \sum \limits_{i = 1}^P {\tilde \alpha_i}\mathbf{a}_{r}\left(M, \phi^r_{RIS,i},\zeta^r_{RIS,i}\right)\mathbf{a}_{t}^H\left(N_t, \psi^t_{BS,i}\right),
\label{eq5}
\end{equation}
\begin{equation}
\mathbf{H}_{r,t}=\sqrt{\frac{MN_r}{L}}\mathop \sum \limits_{i = 1}^L {\tilde \beta_i}\mathbf{a}_{r}\left(N_r, \psi^r_{UE,i}\right)\mathbf{a}_{t}^H\left(M, \phi^t_{RIS,i},\zeta^t_{RIS,i}\right),
\label{eq6}
\end{equation}
where $P$ is the total number of paths between the BS and the RIS, $L$ is the total number of paths between the RIS and the UE scheduled in time slot $t$.
$\tilde \alpha_i$ and $\tilde \beta_i$ denote the complex gain of the $i$-th path.
$\phi^r_{RIS,i}$ and $\zeta^r_{RIS,i}$ represent the azimuth and elevation angles of arrival associated with the RIS, respectively,
while $\phi^t_{RIS,i}$ and $\zeta^t_{RIS,i}$ represent the azimuth and elevation angles of departure associated with the RIS, respectively.
$\psi^t_{BS,i}$ denotes the angle of departure from the BS and $\psi^r_{UE,i}$ denotes the angle of arrival to the scheduled UE.
$\mathbf{a}_r(\cdot)$ and $\mathbf{a}_t(\cdot)$ denote the normalized angle steering vector functions at transmitter and receiver, respectively.
Specifically, for the BS and UEs with an $N$-element ULA, the corresponding angle steering vector is expressed as
\begin{equation}
\mathbf{a}(N,\psi)={\frac{1}{\sqrt{N}}}[1,e^{j{\frac{2 \pi d}{\lambda}}\sin(\psi)}, \cdots ,e^{j{\frac{2 \pi d}{\lambda}}(N-1)\sin(\psi)}],
\label{eq7}
\end{equation}
and for the RIS with the UPA with $M=M_a \times M_b$ reflection elements, the corresponding normalized angle steering vector is expressed as
\begin{equation}
\begin{aligned}
& \mathbf{a}\!\left(M,\phi ,\zeta  \right)\!=\!\frac{1}{\sqrt{M}}[1,\cdots ,{{\text{e}}^{j\frac{2\pi d}{\lambda }\left( \left( {{m}_{a}}\!-\!1 \right)\sin\left( \zeta  \right)\sin\left( \phi  \right)+\left( {{m}_{b}}\!-\!1 \right)\cos\left( \zeta  \right) \right)}} \\
& \cdots,\ {{\text{e}}^{j\frac{2\pi d}{\lambda }\left( \left( {{M}_{a}}-1 \right)\sin\left( \zeta  \right)\sin\left( \phi  \right)+\left( {{M}_{b}}-1 \right)\cos\left( \zeta  \right) \right)}}].
\end{aligned}
\label{eq8}
\end{equation}


\subsection{AoI Definition}\label{S3-4}
We use the AoI to measure the freshness of information received by UEs.
The BS is assumed to follow a \emph{per time slot sampling strategy}, i.e., it samples status-update information and sends a status-update packet at the beginning of each time slot~\cite{27}.
Meanwhile, a single packet queue discipline is considered for the BS, which means the older status-update packet will be replaced by a newly arrived packet.
We use $u_{k,t} \in \{0,1\} $ to indicate whether UE $k$ is scheduled to receive data from the BS in time slot $t$. If UE $k$ is scheduled, $u_ {k,t}=1$, and then the BS sends a status-update packet to UE $k$; otherwise, $u_ {k,t}=0$.
The overall scheduling strategy in time slot $t$ is expressed as
$\mathbf{u}_t=\left[u_ {1,t},u_ {2,t},\ldots,u_ {K,t}\right]^T$.
Note that in addition to being scheduled, the successful transmission of the status-update information to UE $k$ requires that the SNR exceeds the threshold for reliable demodulation.
In case the status-update packet is successfully transmitted to UE $k$ in time slot $t$, the AoI of UE $k$ is reset to 1, otherwise, the AoI is increased by 1.
Therefore, the evolution of the AoI of UE $k$ is given by
\begin{equation}
\mathcal{A}_{k,t} =
\begin{cases}
1, & \text{if $u_{k,t}=1$ and $\gamma_t>\gamma_{th}$,} \\
\mathcal{A}_{k,t-1}+1 , & \text{otherwise.}
 \end{cases}
 \label{eq9}
 \end{equation}

For simplicity, we assume that the initial AoI $\mathcal{A}_{k,0} = 1,\forall k$.
The average AoI of UE $k$ during the transmission phase consisting of $T$ time slots is given by
\begin{equation}
\mathcal{A}_k=\frac{1}{T}\sum_{t=1}^{T}\mathcal{A}_{k,t}.
\label{eq10}
\end{equation}

Considering the requirement of each UE for fresh information, we denote the maximum tolerable AoI for UE $k$ as $\mathcal{A}_{k,\text{max}}$. The AoI of each UE should satisfy
$\mathcal{A}_{k} \le \mathcal{A}_{k,\text{max}}, \forall k$.
In this paper, we focus on the situation in which each UE receives the same types of service from the BS.
Generally, we assume $\mathcal{A}_{k,\text{max}}=\mathcal{A}_{\text{max}}, \forall k$.

\section{Problem Formulation and Decomposition}\label{S4}
In this section, we first formulate our optimization problem
based on the above system model, and then decompose the complex problem in order to efficiently solve it.
\subsection{Sum Rate Maximization Problem Formulation}\label{S4-1}
In this paper, we aim to maximize the sum rate of the system over $T$ time slots
by jointly optimizing the scheduling strategy, the reflection-coefficient matrix of the RIS, and the beamforming vector of the BS and UEs.
To facilitate the subsequent presentation,
let $\mathbf{\Phi}=\left[\mathbf{\Phi}_1,\mathbf{\Phi}_2,\ldots,\mathbf{\Phi}_T\right]^T$ denote the RIS reflection coefficient matrix over $T$ time slots,
$\mathbf{W}=\left[\mathbf{w}_1,\mathbf{w}_2,\ldots,\mathbf{w}_T\right]^T$
and $\mathbf{F}=\left[\mathbf{f}_1,\mathbf{f}_2,\ldots,\mathbf{f}_T\right]^T$
denote the beamforming vector of the BS and UEs over $T$ time slots, respectively,
and $\mathbf{U} = \left[\mathbf{u}_1,\mathbf{u}_2,\ldots,\mathbf{u}_T\right]^T$ represent the scheduling strategy over $T$ time slots.
Then,
the joint optimization problem (P1) can be formulated as
\begin{align}
& \mathop {\max }\limits_{{\bf{U}},{\bf{W}},{\bf{\Phi }},{\bf{F}}} \;\mathop \sum \limits_{t = 1}^T R_t  \label{eq11} \\
s.t.  &\;\; \mathop \sum \limits_{k = 1}^K {u_{k,t}} = 1,\;\forall t \in \left\{ {1,2, \cdots ,T} \right\}, \label{eq12} \\
&\;\; \mathcal{A}_k \le \mathcal{A}_{k,\text{max}},\; \forall k \in \left\{ {1,2, \cdots ,K} \right\}, \label{eq13}\\
&\;\; {u_{k,t}} \in \left\{ {0,1} \right\},\forall k \in \left\{ {1,2, \cdots ,K} \right\},t \in \left\{ {1,2, \cdots ,T} \right\}, \label{eq14}\\
&\;\; {\varphi _{m,t}} \in {\cal F},\;\forall m \in \left\{ {1,2, \cdots ,M} \right\},t \in \left\{ {1,2, \cdots ,T} \right\}, \label{eq15} \\
&\;\; {{{\bf{w}}_t}} \in \Gamma_t,\;\forall t \in \left\{ {1,2, \cdots ,T} \right\}, \label{eq16} \\
&\;\; {{{\bf{f}}_t}} \in \Gamma_r,\;\forall t \in \left\{ {1,2, \cdots ,T} \right\}, \label{eq17}
\end{align}
where Constraint~(\ref{eq12}) indicates that only one UE is scheduled in each time slot, and Constraint~(\ref{eq13})
guarantees that the information freshness of each UE is ensured. Constraint~(\ref{eq14}) limits scheduling variables $u_{k,t}$ to 0-1 variables, and Constraints~(\ref{eq15}-\ref{eq17}) restrict $\varphi _{m,t}$, ${\bf{w}}_t$, and ${\bf{f}}_t$ to be discrete values.
The problem is an integer non-convex optimization problem.
Moreover, the four variables, $\mathbf{U}$, $\mathbf{W}$, $\mathbf{\Phi}$, and $\mathbf{F}$, are coupled in both the objective function and $\mathcal{A}_k$.
Although the global optimal solution can be found by exhaustive search, the multi-variable coupling makes the search space prohibitively large and consequently, the computational overhead considerable.
Facing these challenges, our goal is to design a low-complexity algorithm to solve this problem.

\subsection{Problem Decomposition}\label{S4-2}
Considering the multi-variable coupling, we first decompose the problem. First, it can be noted that only one UE is scheduled in each time slot in the TDMA system. Let $\mathbf{H}_{r,k,t}$ denote the channel between RIS and UE $k$ in time slot $t$. The corresponding transmit beamforming vector, the receive beamforming vector, and the RIS phase shift matrix are represented as $\mathbf{w}_{k,t}$, $\mathbf{f}_{k,t}$, and $\mathbf{\Phi}_{k,t}$, respectively. The sum rate can be accordingly rewritten as
\begin{align}
\mathop \sum \limits_{t = 1}^T \! R_t
&\!=\! \mathop \sum \limits_{t = 1}^T \! \mathop \sum \limits_{k = 1}^K \! u_{k\!,t} {\log _2} \! \left( \! {1 \! + \! \frac{\left|{\mathbf{f}_{k\!,t}^H}{(\mathbf{G}_t\mathbf{\Phi}_{k\!,t}\mathbf{H}_{r\!,k\!,t})}^H\mathbf{w}_{k\!,t}\sqrt{P_T}\right|^2}{\sigma^2}}\! \right) \!
\label{eq:rde}
\end{align}
Then, we assume that the channels are quasi-static and do not change over $T$ time slots, so we have $\mathbf{G}_t = \mathbf{G}, \mathbf{H}_{r,k,t} = \mathbf{H}_{r,k}, \;\forall t \in \left\{ {1,2, \cdots ,T} \right\}$. In this case, the beamforming vectors and the RIS phase shifts for UE $k$ can be simplified to be consistent in different time slots, which can be represented as $\mathbf{w}_{k,t} = \mathbf{w}_{k}$, $\mathbf{f}_{k,t} = \mathbf{f}_{k}$, and $\mathbf{\Phi}_{k,t} = \mathbf{\Phi}_{k}, \;\forall t \in \left\{ {1,2, \cdots ,T} \right\}$.
Accordingly, the sum rate can be further rewritten as
\begin{align}
\mathop \sum \limits_{t = 1}^T R_t \!=\! \mathop \sum \limits_{t = 1}^T \mathop \sum \limits_{k = 1}^K u_{k,t} {\log _2} \left( {1 + \frac{\left|{\mathbf{f}_k^H}{(\mathbf{G}\mathbf{\Phi}_k\mathbf{H}_{r,k})}^H\mathbf{w}_k\sqrt{P_T}\right|^2}{\sigma^2} }\right)
\end{align}
Let $R_k = {\log _2} \left( {1 + \frac{\left|{\mathbf{f}_k^H}{(\mathbf{G}\mathbf{\Phi}_k\mathbf{H}_{r,k})}^H\mathbf{w}_k\sqrt{P_T}\right|^2}{\sigma^2} }\right)$ represent the transmission rate of UE $k$.
Thus, P1 can be rewritten as
\begin{align}
& \mathop {\max }\limits_{{\bf{U}}; {{\bf{w}}_k},{{\bf{\Phi }}_k},{{\bf{f}}_k}, \forall k }\;\mathop \sum \limits_{t = 1}^T \mathop \sum \limits_{k = 1}^K u_{k,t} R_k  \label{eq11b} \\
s.t.  &\;\; \mathop \sum \limits_{k = 1}^K {u_{k,t}} = 1,\;\forall t \in \left\{ {1,2, \cdots ,T} \right\}, \label{eq12b} \\
&\;\; \mathcal{A}_k \le \mathcal{A}_{k,\text{max}},\; \forall k \in \left\{ {1,2, \cdots ,K} \right\}, \label{eq13b}\\
&\;\; {u_{k,t}} \in \left\{ {0,1} \right\},\forall k \in \left\{ {1,2, \cdots ,K} \right\},t \in \left\{ {1,2, \cdots ,T} \right\}, \label{eq14b}\\
&\;\; {\varphi _{m,k}} \in {\cal F},\;\forall m \in \left\{ {1,2, \cdots ,M}\right\}, \forall k \in \left\{ {1,2, \cdots ,K} \right\},  \label{eq15b} \\
&\;\; {{{\bf{w}}_k}} \in \Gamma_t, \forall k \in \left\{ {1,2, \cdots ,K} \right\},\label{eq16b} \\
&\;\; {{{\bf{f}}_k}} \in \Gamma_r, \forall k \in \left\{ {1,2, \cdots ,K} \right\}. \label{eq17b}
\end{align}

Note that the sum rate can also be converted as $\mathop \sum \limits_{t = 1}^T \mathop \sum \limits_{k = 1}^K u_{k,t} R_k = \mathop \sum \limits_{k = 1}^K (\mathop \sum \limits_{t = 1}^T  u_{k,t}) R_k$, where $\mathop \sum \limits_{t = 1}^T  u_{k,t} \geq 0, \forall k \in \left\{ {1,2, \cdots ,K}\right\}$. Since the transmission rate $R_k$ is related only to ${\bf{w}}_k$, ${{\bf{\Phi }}_k}$ and ${{\bf{f}}_k}$ and not to ${\bf{U}}$, and the transmission rates of different UEs are independent of each other, we can decompose P1 into $K$ \emph{per-UE rate maximization problems} and a \emph{scheduling strategy design problem}.
\smallskip
\subsubsection{\textbf{Per-UE rate maximization problem}}
This subproblem aims to maximize the achievable transmission rate of each UE through the joint optimization of beamforming vectors and the reflection-coefficient matrix of the RIS. We denote the achievable transmission rate of UE $k$ as $R_k$, and the subproblem can be written as
\begin{align}
& \mathop {\max }\limits_{{{\bf{w}}_k},{{\bf{\Phi }}_k},{{\bf{f}}_k}} \;{R_k} = {\log _2}\left( {1 + \frac{{{{\left| {{{\bf{f}}_k}^H{{\bf{H}}_k}^H{{\bf{w}}_k}\sqrt {{P_T}} } \right|}^2}}}{{{\sigma ^2}}}} \right)  \label{eq18} \\
s.t.
&\qquad\qquad {\varphi _{m,k}} \in {\cal F},\;\forall m \in \left\{ {1,2, \cdots ,M} \right\}, \label{eq19} \\
&\qquad\qquad {{{\bf{w}}_k}} \in \Gamma_t, \label{eq20} \\
&\qquad\qquad {{{\bf{f}}_k}} \in \Gamma_r. \label{eq21}
\end{align}
In this subproblem, variables $\mathbf{w}_k$, $\mathbf{\Phi}_k$, and $\mathbf{f}_k$ are still coupled, so we further decompose the subproblem into a \emph{beamforming optimization subproblem} and a \emph{RIS reflection coefficient optimization subproblem}.

For the \emph{beamforming optimization subproblem}, we assume that the reflection coefficient matrix of the RIS ${\bf{\Phi }}_k$ is fixed and maximize the transmission rate of each UE by optimizing the beamforming vectors $\mathbf{w}_k$ and $\mathbf{f}_k$. We can write this subproblem as
\begin{align}
& \mathop {\max }\limits_{{{\bf{w}}_k},{{\bf{f}}_k}} \;{R_k} = {\log _2}\left( {1 + \frac{{{{\left| {{{\bf{f}}_k}^H{{\bf{H}}_k}^H{{\bf{w}}_k}\sqrt {{P_T}} } \right|}^2}}}{{{\sigma ^2}}}} \right)  \label{eq22} \\
s.t.
&\quad\quad\;\;\; {{{\bf{w}}_k}} \in \Gamma_t, \label{eq23} \\
&\quad\quad\;\;\; {{{\bf{f}}_k}} \in \Gamma_r. \label{eq24}
\end{align}

For the \emph{RIS reflection coefficient optimization subproblem}, we fix the beamforming vectors and find the efficient reflection coefficient matrix ${\bf{\Phi }}_k$ for the RIS. We can write this subproblem as
\begin{align}
& \mathop {\max }\limits_{{{\bf{\Phi }}_k}} \;{R_k} = {\log _2}\left( {1 + \frac{{{{\left| {{{\bf{f}}_k}^H{{\bf{H}}_k}^H{{\bf{w}}_k}\sqrt {{P_T}} } \right|}^2}}}{{{\sigma ^2}}}} \right)  \label{eq25} \\
s.t.
&\;\qquad\; {\varphi _{m,k}} \in {\cal F},\;\forall m \in \left\{ {1,2, \cdots ,M} \right\}. \label{eq26}
\end{align}

\smallskip
\subsubsection{\textbf{Scheduling strategy design problem}}
Based on the maximum achievable transmission rate of each UE, this subproblem is to design the scheduling strategy to maximize the total transmission rate over $T$ time slots.
The subproblem can be written as
\begin{align}
& \mathop {\max }\limits_{\mathbf{U}}\sum_{t=1}^{T}\sum_{k=1}^{K}{u_{k,t}R_k}  \label{eq27} \\
s.t. &\;\; \mathop \sum \limits_{k = 1}^K {u_{k,t}} = 1,\;\forall t \in \left\{ {1,2, \cdots ,T} \right\}, \label{eq28} \\
&\;\; \mathcal{A}_k \le \mathcal{A}_{k,\text{max}},\; \forall k \in \left\{ {1,2, \cdots ,K} \right\}, \label{eq29}\\
&\;\; {u_{k,t}} \in \left\{ {0,1} \right\},\forall k \in \left\{ {1,2, \cdots ,K} \right\},t \in \left\{ {1,2, \cdots ,T} \right\}. \label{eq30}
\end{align}

In the next section, we will develop effective algorithms to solve the two decomposed problems and achieve the goal of sum rate maximization.

\section{Sum rate maximization}\label{S5}
In this section, we aim to propose a low-complexity algorithm to solve P1. Based on the decomposition of P1, the proposed solution consists of three parts: First, in Section V-A, we design a block coordinate descent (BCD)-based algorithm to solve the \emph{per-UE rate maximization problem}, which solves the \emph{beamforming optimization subproblem} and the \emph{RIS reflection coefficient optimization subproblem} iteratively until the algorithm converges. Then, in Section V-B, we propose a heuristic scheduling algorithm to solve the \emph{scheduling strategy design problem}. Finally, in Section V-C, we show the overall sum rate maximization algorithm for solving P1. The convergence and complexity analyses are given in Section V-D.

\subsection{Per-UE Rate Maximization}\label{S5-1}
To solve the \emph{per-UE rate maximization problem}, we first design algorithms to solve the \emph{beamforming optimization subproblem} and the \emph{RIS reflection coefficient optimization subproblem}. Then we use the BCD algorithm to obtain the overall suboptimal solution. It is worth noting that we consider the difficulty of channel estimation in the RIS-assisted mmWave MIMO system. Thus, different from most existing studies adopting the BCD-based method~\cite{07,34}, we design the BCD algorithm in the case of unknown CSI.

\smallskip
\subsubsection{\textbf{Beamforming Optimization}}

Given the codebook of the BS and UE, although exhaustively searching all the transmit-receive beam pairs in the codebooks can find an efficient beam pair, we choose the hierarchical search method for reduced complexity.
Specifically, we first design multilevel codebooks with different beam widths and then perform a divide-and-conquer search on the different codebook levels. The hierarchical search shows a tree structure, thereby substantially enhancing the search efficiency.
The details of this method are given by Algorithm~\ref{alg.a}.

First, we focus on the design of the hierarchical codebooks $\Gamma_t$ and $\Gamma_r$.
There are two criteria to design a hierarchical codebook~\cite{36}, which are given as follows.
\begin{itemize}
  \item Within each layer, the aggregate beam coverage of all codewords should span the entirety of the angular domain, which ensures no miss of any angle during the beam search.
  \item The beam coverage of an arbitrary codeword within a layer should be given by the union of those of several adjacent codewords in the next layer, which establishes a tree-fashion relationship between the codewords.
\end{itemize}

In this paper, we assume that each parent codeword has 2 child codewords, thus forming a binary-tree codebook structure. Fig.~\ref{fig:2} shows a three-layer binary-tree codebook structure diagram, where $\mathbf{w}(l,n)$ denotes the $n$-th codeword of the $l$-th layer codebook. For antenna arrays with $N$ antenna elements, we assume that there are $N$ codewords covering the angle range [-1,1] in the last layer and each codeword has beam width $2/N$ with different steering angles. Therefore, the codebook consists of $\log_2(N)+1$ layers, where the $k$-th layer consists of $2^k$ codewords with beam width $2/2^k$ for each codeword.

For the design of $N$ codewords in the last layer,
since the steering vector $\mathbf{a}(N,\psi)$ in~(\ref{eq7}) can be defined as having a $2/N$ beam width centered on the steering angle $\psi$,
we adopt the steering vectors with $N$ angles evenly sampled within [-1,1]~\cite{36}. The $n$-th codeword exhibits the maximal beam gain along the angle $-1+\frac{2n-1}{N}$. Thus, we have $\mathbf{w}(\log_2(N),n)= \mathbf{a}(N,-1+\frac{2n-1}{N}),\; n=1,2,\cdots,N$.
Then, for the design of codewords in the other layers, we use the joint sub-array and deactivation approach~\cite{36}. Specifically, to broaden the beam, we divide the $N$-antenna array into $Q$ sub-arrays. Each sub-array is equipped with $N_S$ antennas.
Taking the first codeword of each layer (i.e., $\mathbf{w}(l,1)$) as an example, the codeword of the $q$-th sub-array can be represented as $\mathbf{w}_q = [\mathbf{w}(l,1)]_{(q-1)N_S+1:qN_S}$. Among these sub-arrays, the number of the activated sub-arrays is denoted by $N_A$. Since the beams of these activated sub-arrays are pointed in sufficiently spaced directions, they can be aggregated into wider beams.
We define the codeword of the $q$-th activated sub-array as $\mathbf{w}_q = {e}^{j\theta_q} \mathbf{a}(N_S,-1+\frac{2q-1}{N_S})$ with ${e}^{j\theta_q}$ representing a scalar coefficient with the unit norm for the $q$-th sub-array. To reduce beam fluctuations, the intersection points between each sub-array coverage area are required to have high beam gain, which is modeled as the problem (27) in~\cite{36}. Based on the solution for the problem, we have $\theta_q = -q\frac{{{N}_{S}}-1}{{{N}_{S}}}\pi$. For the deactivated sub-arrays, the antennas in these sub-arrays are turned off, i.e., $\mathbf{w}_q = \mathbf{0}_{N_S\times 1}$. Thus, the codeword of each sub-array can be given by
\begin{equation}
\mathbf{w}_q =
\begin{cases}
&{{e}^{-jq\frac{{{N}_{S}}-1}{{{N}_{S}}}\pi }}\mathbf{a}\left( {{N}_{S}},-1+\frac{2q-1}{{{N}_{S}}} \right),\; q=1,2,...,{{N}_{A}} \\
&{\mathbf{0}_{N_S\times 1}}, \;q={{N}_{A}}+1,{{N}_{A}}+2,...,Q,
\end{cases}
 \label{eq31}
\end{equation}
and the beam width of the sub-array codeword is $\frac {2N_A}{N_S}$. In addition, according to Corollary 1 in~\cite{36}, after obtaining the first codeword of
each layer (i.e., $\mathbf{w}(l,1)$), we can obtain all the other codewords in the same layer through rotating $\mathbf{w}(l,1)$ by $\frac{2(n-1)}{2^l}$, $n=2,3,...,{2}^{l}$, respectively. The beam rotation can be realized by
\begin{equation}
\mathbf{w}(l,n) = \mathbf{w}(l,1)\circ  \sqrt{N}\mathbf{a}(N,\frac{2(n-1)}{2^l}),n=2,3,\cdots,2^l
 \label{eq31a}
\end{equation}
where $\circ$ represents entry-wise product.

\begin{algorithm}[htbp]
\small
\caption{\small Hierarchical Search Method for Beamforming Optimization}	
\label{alg.a}
\begin{algorithmic}[1]
\renewcommand{\algorithmicrequire}{ \textbf{Input:}}
\REQUIRE{$N_t$; $N_r$; the designed hierarchical codebooks $\Gamma_t$ and $\Gamma_r$}
\renewcommand{\algorithmicrequire}{ \textbf{Output:}}
\REQUIRE{$\mathbf{w}_k$; $\mathbf{f}_k$}
\renewcommand{\algorithmicrequire}{ \textbf{Initialization:}}
\REQUIRE{$cw_t=cw_r=0$}
\STATE Fix the BS to be in an omni-directional mode;
\FOR{each layer $l_r$ in $\Gamma_r$}
\STATE Compare the data rate of the $(2cw_r-1)$-th codeword with that of the $(2cw_r)$-th codeword and record the index of the codeword with a higher data rate as $cw_r^*$;
\STATE $cw_r = cw_r^*$;
\ENDFOR
\STATE $\mathbf{f}_k = \Gamma_r(L_r,cw_r)$;
\STATE Fix the UE to the directional mode with $\mathbf{f}_k$;
\FOR{each layer $l_t$ in $\Gamma_t$}
\STATE Compare the data rate of the $(2cw_t-1)$-th codeword with that of the $(2cw_t)$-th codeword and record the index of the codeword with a higher data rate as $cw_t^*$;
\STATE $cw_t = cw_t^*$;
\ENDFOR
\STATE $\mathbf{w}_k = \Gamma_t(L_t,cw_t)$;
\end{algorithmic}
\end{algorithm}

The details of the codebook design are presented in Algorithm~\ref{alg.b}.
First, for the last layer of the codebook, the steering vectors with $N$ angles evenly sampled within [-1,1] are used for codewords as in lines 2-3. Next, the joint sub-array and deactivation approach is adopted to generate the codewords for other layers as in lines 5-14. We first separate $\mathbf{w}(l,1)$ into $Q = 2^{\lfloor(p+1)/2 \rfloor}$ sub-arrays with $ p = \log_2(N)-l$ in lines 5-6 and determine whether to activate half or all of the sub-arrays based on the parity of $p$ in lines 7-11. Then, in line 12, the codebook for each sub-array can be obtained by~(\ref{eq31}) and we can get $\mathbf{w}(l,1)$ accordingly. In the end, based on $\mathbf{w}(l,1)$, we can derive all the other codewords in each layer by~(\ref{eq31a}) in line 13. The hierarchical codebook design is finished after normalizing $\mathbf{w}(l,n)$ in line 14.

\begin{algorithm}[htbp]
\small
\caption{\small Hierarchical Codebook Design }	
\label{alg.b}
\begin{algorithmic}[1]
\FOR{each layer $l$}
\IF{$l=\log_2(N)$}
\STATE $\mathbf{w}(l,n)= \mathbf{a}(N,-1+\frac{2n-1}{N}),\; n=1,2,\cdots,N$;
\ELSE
\STATE $ p = \log_2(N)-l$;
\STATE Separate $\mathbf{w}(l,1)$ into $Q = 2^{\lfloor(p+1)/2 \rfloor}$ sub-arrays with $\mathbf{w}_q = [\mathbf{w}(l,1)]_{(q-1)N_S+1:qN_S},\; q=1,2,\cdots,Q$;
\IF{$p$ is odd}
\STATE $N_A = Q/2$;
\ELSE
\STATE $N_A = Q$;
\ENDIF
\STATE Calculate $\mathbf{w}_q$ according to~(\ref{eq31}) for $q=1,2,\cdots,Q$ and obtain $\mathbf{w}(l,1)$;
\STATE Obtain all the other codewords in layer $l$ through~(\ref{eq31a});
\STATE Normalize $\mathbf{w}(l,n)$;
\ENDIF
\ENDFOR
\end{algorithmic}
\end{algorithm}

After obtaining the codebooks $\Gamma_t$ and $\Gamma_r$, we first fix the BS in the omnidirectional mode and perform a binary tree search in $\Gamma_r$ to find the efficiently received codeword for the UE, as in lines 1-6 in Algorithm~\ref{alg.a}. Specifically, in each layer, we select the codeword with a higher data rate, and the two adjacent codewords in the next-layer codebook within the beam coverage of this codeword are used as candidate codewords for the choice of the next layer.
Then we fix the UE in the directional mode corresponding to the codeword and perform the same binary tree search in $\Gamma_t$ to find the efficient transmit codeword for the BS in lines 7-12.

\begin{figure}[!t]
\centering
\includegraphics*[width=3.2in]{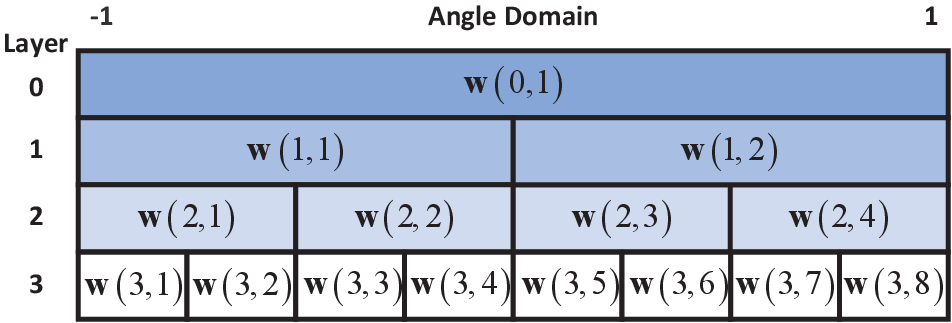}
\centering
\caption{Beam coverage of a 3-layer codebook.}
\label{fig:2}
\end{figure}

\smallskip
\subsubsection{\textbf{RIS Reflection Coefficient Optimization}}
For the RIS reflection coefficient optimization subproblem, we need to select the appropriate phase shift for each RIS element from a finite set of discrete phase shifts. Considering the complexity, we will use the local search method to solve the subproblem as shown in Algorithm~\ref{alg.c}.
Specifically, we optimize each RIS element successively while keeping the phase shifts of the remaining $M-1$ elements fixed.
For each element, we traverse all the possible phase shifts and select the phase shift giving the maximum UE transmission rate as the optimized phase shift for the element.
Then we use it for the phase shift optimization of other RIS elements until all the phase shifts are optimized.

\begin{algorithm}[htbp]
\small
\caption{\small Local Search Method for RIS Reflection Coefficient Optimization}	
\label{alg.c}
\begin{algorithmic}[1]
\renewcommand{\algorithmicrequire}{ \textbf{Input:}}
\REQUIRE{$M$; $b$}
\renewcommand{\algorithmicrequire}{ \textbf{Output:}}
\REQUIRE{$\mathbf{\Phi}_k$}
\FOR{$m=1:M$}
\STATE $R_k^*=0$;
\FOR{$p_s=1:2^b$}
\STATE Update $\mathbf{\Phi}_k$ with $\varphi _{m,k} = (p_s-1) \frac{2\pi}{2^b-1}$;
\STATE Obtain the transmission rate $R_k$;
\IF{$R_k>R_k^*$}
\STATE $R_k^* = R_k$, $\varphi _{m,k}^* = \varphi _{m,k}$;
\ENDIF
\ENDFOR
\STATE Update $\mathbf{\Phi}_k$ with $\varphi_{m,k}^*$;
\ENDFOR
\end{algorithmic}
\end{algorithm}

\smallskip
\subsubsection{\textbf{Joint Optimization for Per-UE Rate Maximization}}
In order to solve the per-UE rate maximization problem, we apply the BCD method to alternately optimize the beamforming vectors and the RIS phase shift matrix.
Specifically, as Algorithm~\ref{alg.d} shows, we randomly initialize the beamforming vectors and the RIS phase shift matrix in the beginning.
In each iteration, we first fix the RIS phase shift matrix to the last updated value and use Algorithm~\ref{alg.a} to update beamforming vectors.
If the data rate of UE $k$ after the beamforming update is more than that before this update, the results of this update are retained; otherwise, the beamforming vectors are not updated.
Then, we update the RIS phase shift matrix based on  Algorithm~\ref{alg.c} with updated beamforming vectors.
If the ratio of the difference in the data rate between two consecutive iterations is less than a certain threshold, i.e., ${|R_k^{\tau}-R_k^{\tau-1}|}/{R_k^{\tau-1}}<\delta$, we consider the algorithm has converged.

\begin{algorithm}[htbp]
\setstretch{1.2}
\small
\caption{\small BCD Method for Joint Optimization of Beamforming and RIS Phase Shift}	
\label{alg.d}
\begin{algorithmic}[1]
\renewcommand{\algorithmicrequire}{ \textbf{Input:}}
\REQUIRE{$M$; $b$; $N_t$; $N_r$; the designed hierarchical codebooks $\Gamma_t$ and $\Gamma_r$}
\renewcommand{\algorithmicrequire}{ \textbf{Output:}}
\REQUIRE{$\mathbf{\Phi}_k^*$; $\mathbf{w}_k^*$; $\mathbf{f}_k^*$; $R_k^*$; $\gamma_k^*$}
\renewcommand{\algorithmicrequire}{ \textbf{Initialization:}}
\REQUIRE{$\tau=0$; $R_k^0=R_{bf}=0$; $\delta=3\times10^{-3}$; randomly generate $\mathbf{\Phi}_k^0$, $\mathbf{w}_k^0$ and $\mathbf{f}_k^0$; $\mathbf{\Phi}_k^* = \mathbf{\Phi}_k^0$; $\mathbf{w}_k^* = \mathbf{w}_k^0$; $\mathbf{f}_k^* = \mathbf{f}_k^0$}
\REPEAT
\STATE Obtain $\mathbf{f}_k^{\tau+1}$ and $\mathbf{w}_k^{\tau+1}$ with fixed $\mathbf{\Phi}_k^\tau$ using Algorithm~\ref{alg.a};
\STATE Obtain data rate $R_{bf}$ with $\mathbf{\Phi}_k^\tau$, $\mathbf{f}_k^{\tau+1}$ and $\mathbf{w}_k^{\tau+1}$;
\IF{$R_k^{\tau} > R_{bf}$}
\STATE $\mathbf{f}_k^{\tau+1} = \mathbf{f}_k^{\tau}$, $\mathbf{w}_k^{\tau+1} = \mathbf{w}_k^{\tau}$;
\ENDIF
\STATE Obtain $\mathbf{\Phi}_k^{\tau+1}$ with fixed $\mathbf{f}_k^{\tau+1}$ and $\mathbf{w}_k^{\tau+1}$ using Algorithm~\ref{alg.c};
\STATE Obtain $R_k^{\tau+1}$ and SNR $\gamma_k^{\tau+1}$ with $\mathbf{\Phi}_k^{\tau+1}$, $\mathbf{f}_k^{\tau+1}$ and $\mathbf{w}_k^{\tau+1}$;
\STATE Update $\tau = \tau+1$;
\UNTIL ${|R_k^{\tau}-R_k^{\tau-1}|}/{R_k^{\tau-1}}<\delta$
\STATE Update $R_k^* = R_k^{\tau}$, $\gamma_k^* = \gamma_k^{\tau}$, $\mathbf{\Phi}_k^* = \mathbf{\Phi}_k^{\tau}$, $\mathbf{w}_k^* = \mathbf{w}_k^{\tau}$, $\mathbf{f}_k^* = \mathbf{f}_k^{\tau}$.
\end{algorithmic}
\end{algorithm}

\subsection{Scheduling Strategy Design}\label{S5-2}
\subsubsection{Motivation and Main Idea}\label{S5-2-1}

After obtaining the maximum achieved rate of each UE with Algorithm~\ref{alg.d},
we need to design a scheduling strategy $\mathbf{U}$ to solve the scheduling strategy design problem. The difficulty in solving this problem lies in how to maximize the system sum rate while meeting the information freshness requirement of each UE. It should be noted that only when the SNR exceeds the threshold for reliable demodulation, can the AoI be reduced. Scheduling the UE with SNR below the threshold will not contribute to the information freshness requirement satisfaction and the sum rate enhancement. Thus, we filter UEs based on SNR and only schedule UEs with SNR above the threshold.

To schedule these filtered UEs, we design two scheduling phases. First, in \emph{Scheduling Phase I}, we wish to ensure that the information freshness requirement of each UE is satisfied. Since we assume each UE receives the same types of service from the BS, which typically means the same information freshness requirements, we adopt a uniform and fair scheduling strategy. We schedule each UE in turn in the descending order of data rate. Each of the $K$ UEs is scheduled once every $K$ timeslots. Within a limited time $T$, this uniform and fair scheduling strategy can ensure that the time slot interval between two adjacent scheduling time slots is consistent for each UE, the difference in the number of scheduling time slots between UE with maximum rate and UE with minimum rate is not greater than once, and the maximum AoI of each UE over $T$ timeslots will not exceed the number of UEs $K$. Therefore, each UE has similar AoI performance. Besides, according to Proposition 1 in~\cite{28}, this uniform and fair scheduling strategy achieves the lower bound of the average episodic AoI, which is defined as $\frac{1}{K}\sum_{k=1}^{K}\mathcal{A}_k$. Motivated by these facts, we use the fair scheduling strategy to obtain better information freshness guarantees.

Then, in \emph{Scheduling Phase II}, we wish to enhance the system sum rate as much as possible based on the scheduling result in Phase I. The main idea is to schedule UEs with the highest data rate as many times as possible without violating the AoI constraint of other UEs. Specifically, we select the UE with the highest data rate as the target UE and traverse all time slots. In each time slot, we replace the scheduled UE with the target UE and test the AoI of the originally scheduled UE. Only if the AoI satisfies the constraint~(\ref{eq29}), is the current replacement adopted.


\subsubsection{Heuristic Scheduling Algorithm}\label{S5-2-2}
The pseudocode of the heuristic scheduling algorithm is presented in Algorithm~\ref{alg.e}.
For ease of presentation, we use $\mathcal{R}$ to denote the set of maximum achievable data rates for UEs without violating SNR constraints. The set of UEs without violating SNR constraints is represented by $\mathcal{K}_u$. The mapping between UE $k$ and its data rate $R_k^*$ is denoted by $k=\mathfrak{K}(R_k^*)$ and $K_t$ is used to represent the scheduled UE in time slot $t$. First, we allocate time slots to UEs that can reliably demodulate the received signal in descending order of their data rates in \emph{Scheduling Phase I}, which corresponds to lines 1-8. Specifically, in each time slot, we schedule the UE with the highest data rate in set $\mathcal{R}_t$, as in lines 2-3. In line 4, the data rate of the scheduled UE is removed from $\mathcal{R}_t$. In lines 5-7, if $\mathcal{R}_t$ is an empty set, which means that all UEs have been scheduled for one round, reinitialize $\mathcal{R}_t$ to $\mathcal{R}$ for the next round of scheduling. Next, we adjust the scheduling strategy to maximize the system sum rate as much as possible in \emph{Scheduling Phase II}, as in lines 9-20. Specifically, we denote the UE with the maximum achievable data rate as $k^{\text{max}}$. For each timeslot $t$, if UE $k^{\text{max}}$ is not scheduled, we replace the scheduled UE $K_t$ to UE $k^{\text{max}}$, as in lines 11 and 12. In lines 13-18, we calculate the AoI of UE $K_t$ and determine whether the AoI constraint is satisfied. If yes, then the adjustment is applied; otherwise, it is not applied. After completing the adjustments of all the timeslots, we obtain the final scheduling strategy.

\begin{algorithm}[htbp]
\small
\caption{\small Scheduling Strategy Design}	
\label{alg.e}
\begin{algorithmic}[1]
\renewcommand{\algorithmicrequire}{ \textbf{Input:}}
\REQUIRE{$\mathcal{R}$; $\mathcal{K}_u$; $T$; $\mathcal{A}_{k,\text{max}}, \forall k$}
\renewcommand{\algorithmicrequire}{ \textbf{Output:}}
\REQUIRE{$\mathbf{U}^*$}
\renewcommand{\algorithmicrequire}{ \textbf{Initialization:}}
\REQUIRE{$u_{k,t} = 0, \forall k,t $; $K_t = 0, \forall t $; $\mathcal{R}_1 = \mathcal{R}$}
\FOR {each time slot $t$}
\STATE $\bar k = \mathfrak{K}(\mathrm{max}(\mathcal{R}_t))$;
\STATE $u_{\bar k,t} = 1$, $K_t = \bar k$;
\STATE $\mathcal{R}_{t+1} = \mathcal{R}_t-\{R_{\bar k}^*\}$;
\IF{$\mathcal{R}_{t+1} = \emptyset$}
\STATE $\mathcal{R}_{t+1} = \mathcal{R}$;
\ENDIF
\ENDFOR
\STATE Treat the UE in $\mathcal{K}_u$ with the highest achievable data rate as the target UE, which is denoted as $k^{\text{max}}$.
\FOR {each time slot $t$} 
\IF{$u_{k^{\text{max}},t} = 0$} 
\STATE $u_{k^{\text{max}},t} = 1, u_{K_t,t} = 0$;
\STATE Calculate the AoI of UE $K_t$;
\IF{$\mathcal{A}_{K_t} \textgreater \mathcal{A}_{K_t,\text{max}}$}
\STATE $u_{k^{\text{max}},t} = 0, u_{K_t,t} = 1$;
\ELSE
\STATE $K_t = k^{\text{max}}$;
\ENDIF
\ENDIF
\ENDFOR  
\end{algorithmic}
\end{algorithm}


\subsection{Sum Rate Maximization}\label{S5-3}
Following the design of the algorithms for the decomposed problems, we propose the sum rate maximization algorithm to solve P1.
The algorithm is given by Algorithm~\ref{alg.f}.
First, we generate the hierarchical codebooks for BS and each UE by Algorithm~\ref{alg.b}. Then we use Algorithm~\ref{alg.d} to obtain the maximum achievable rate and the optimized variables for each UE. As in lines 5-8, we check the SNR of each UE to find the UEs which can achieve successful demodulation. The set $\mathcal{R}$ and $\mathcal{K}_u$ are obtained as the input of Algorithm~\ref{alg.e} to get the scheduling strategy. In the end, since we have obtained $\mathbf{\Phi}_k^*$, $\mathbf{w}_k^*$, $\mathbf{f}_k^*$ for each UE $k$, we can easily obtain $\mathbf{W}^*$, $\mathbf{\Phi}^*$ and $\mathbf{F}^*$ according to the scheduling strategy $\mathbf{U}^*$.

\begin{algorithm}[htbp]
\small
\caption{\small Sum Rate Maximization}	
\label{alg.f}
\begin{algorithmic}[1]
\renewcommand{\algorithmicrequire}{ \textbf{Input:}}
\REQUIRE{$M$; $b$; $K$; $T$; $N_t$; $N_r$; $\mathcal{A}_{k,\text{max}}, \forall k$; $\gamma_{th}$}
\renewcommand{\algorithmicrequire}{ \textbf{Output:}}
\REQUIRE{$\mathbf{U}^*$; $\mathbf{W}^*$; $\mathbf{\Phi}^*$; $\mathbf{F}^*$}
\renewcommand{\algorithmicrequire}{ \textbf{Initialization:}}
\REQUIRE{$\mathcal{R} = \emptyset$; $\mathcal{K}_u = \emptyset$}
\STATE Generate the hierarchical codebook $\Gamma_t$ for BS using Algorithm~\ref{alg.b};
\FOR{each UE $k$}
\STATE Generate the hierarchical codebook $\Gamma_r$ for UE $k$ using Algorithm~\ref{alg.b};
\STATE Obtain $\mathbf{\Phi}_k^*$, $\mathbf{w}_k^*$, $\mathbf{f}_k^*$, $R_k^*$ and $\gamma_k^*$ using Algorithm~\ref{alg.d};
\IF {SNR $\gamma_k^*>\gamma_{th}$}
\STATE $\mathcal{R} = \mathcal{R} \cup \{R_k^*\}$;
\STATE $\mathcal{K}_u = \mathcal{K}_u \cup \{k\}$;
\ENDIF
\ENDFOR
\STATE Obtain $\mathbf{U}^*$ using Algorithm~\ref{alg.e};
\STATE Obtain $\mathbf{W}^*$, $\mathbf{\Phi}^*$ and $\mathbf{F}^*$ based $\mathbf{U}^*$;
\end{algorithmic}
\end{algorithm}

\subsection{Algorithm Analysis}\label{S5-4}

\smallskip
\subsubsection{\textbf{Convergence Analysis}}
In Algorithm~\ref{alg.f}, we can see that the number of iterations for the scheduling algorithm (i.e., Algorithm~\ref{alg.e}) is fixed, while the number of iterations required to obtain the maximum achievable data rate for each UE (i.e., Algorithm~\ref{alg.d}) is uncertain.
Thus, we focus on the convergence of Algorithm~\ref{alg.d}.

First, Theorem 1 shows that the objective function of the original optimization problem is non-decreasing in Algorithm~\ref{alg.d}.
\begin{theorem}
In Algorithm~\ref{alg.d}, $R(\mathbf{\Phi}_k^{\tau+1}, \mathbf{w}_k^{\tau+1}, \mathbf{f}_k^{\tau+1}) \ge R(\mathbf{\Phi}_k^\tau, \mathbf{w}_k^{\tau}, \mathbf{f}_k^{\tau})$.
\end{theorem}
\begin{proof}
In line 2 of Algorithm~\ref{alg.d}, with $\mathbf{\Phi}_k^\tau$ given in the $\tau$-th iteration, the beamforming vectors are obtained by Algorithm~\ref{alg.a}.
It is worth noting that for Algorithm~\ref{alg.a}, an early search stage with weak beamforming gains is likely to experience relatively low SNR.
This may lead to a higher probability of failing to find the best beam pair in the early search phase, resulting in subsequent misalignment at higher levels~\cite{37}. Considering this, in lines 4-6 of Algorithm~\ref{alg.d}, we compare $R_k^{\tau}$ and $R_{bf}$ and decide whether to adopt the results of the hierarchical search. Therefore, $R(\mathbf{\Phi}_k^\tau, \mathbf{w}_k^{\tau+1}, \mathbf{f}_k^{\tau+1}) \ge R(\mathbf{\Phi}_k^\tau, \mathbf{w}_k^{\tau}, \mathbf{f}_k^{\tau})$.
Then, in line 7 of Algorithm~\ref{alg.d}, $\mathbf{\Phi}_k$ is updated using Algorithm~\ref{alg.c} with the beamforming vectors fixed.
The local search algorithm aims at maximizing the sum rate and
searches for better phase shift values for each RIS element on the basis of $\mathbf{\Phi}_k^\tau$.
Therefore, the performance of $\mathbf{\Phi}_k^{\tau+1}$ output by Algorithm~\ref{alg.d} is better than or equal to the performance of $\mathbf{\Phi}_k^{\tau}$, which can be expressed as $R(\mathbf{\Phi}_k^{\tau+1}, \mathbf{w}_k^{\tau+1}, \mathbf{f}_k^{\tau+1}) \ge R(\mathbf{\Phi}_k^\tau, \mathbf{w}_k^{\tau+1}, \mathbf{f}_k^{\tau+1})$.
Thus, $R(\mathbf{\Phi}_k^{\tau+1}, \mathbf{w}_k^{\tau+1}, \mathbf{f}_k^{\tau+1}) \ge R(\mathbf{\Phi}_k^\tau, \mathbf{w}_k^{\tau}, \mathbf{f}_k^{\tau})$. This completes the proof.
\end{proof}

In addition, the number of discrete phase shifts and codewords for transmit and receive beamforming are limited, and the scheduling parameter is 0-1 variables, which makes the problem of maximizing the sum rate bounded and the output solutions guaranteed. Therefore, we have completed the proof of the convergence of the sum rate maximization algorithm.

%

\smallskip
\subsubsection{\textbf{Complexity Analysis}}
Since Algorithm~\ref{alg.f} contains two parallel parts: the per-UE rate maximization and the scheduling strategy design, we analyze their complexity separately.
First, for the per-UE rate maximization in lines 1-9, the $for$ loop in line 2 has $K$ iterations.
In line 4, the complexity of Algorithm~\ref{alg.d} is not only related to the number of iterations for the BCD method, which can be represented as $N_{outer}$ to achieve the convergence condition ${|R_k^{\tau}-R_k^{\tau-1}|}/{R_k^{\tau-1}}<\delta$, but also related to the complexity of the beamforming optimization subproblem and RIS reflection coefficient optimization subproblem.
For the former, two codewords are searched for at each layer of codebooks at both BS and UE.
The complexity of the hierarchical search method is $O(2\log_2N_t+2\log_2N_r)$.
For the latter one, the local search algorithm selects the best one among $2^b$ phase shifts for each element while keeping the phase shifts of the remaining elements unchanged. Since the RIS contains $M$ elements, the complexity of this part is $O(M*2^b)$. Therefore, we get the complexity of Algorithm~\ref{alg.d} as $O(N_{outer}*(2\log_2N_t+2\log_2N_r+M*2^b))$, and the complexity of the per-UE rate optimization is $O(K*(N_{outer}*(2\log_2N_t+2\log_2N_r+M*2^b)))$. Then, The scheduling strategy design corresponds to Algorithm~\ref{alg.e}. Both the $for$ loops in line 1 and line 10 have $T$ iterations, and
The two $for$ loops are parallel.
Therefore, the complexity of Algorithm~\ref{alg.e} is $O(T)$.
In summary, the complexity of the sum rate maximization algorithm is
$O(\max(K*(N_{outer}*(2\log_2N_t+2\log_2N_r+M*2^b)),T))$.
By simulation tests, $N_{outer}$ ranges from 3 to 6.
Such low complexity of the algorithm
makes it suitable for practical implementation.

\section{Simulation Results and Discussions}\label{S6}
In this section, we evaluate the performance of the proposed
algorithm under various representative parameters. We also compare the performance of the proposed
scheme with several baseline schemes and investigate the impact of different parameters on system performance.

\begin{table}[!t]
\caption{Simulation Parameters}
\label{Simulation}
\centering
\small
\setlength{\tabcolsep}{2mm}{
\begin{tabular}{ll}
\toprule 
Parameter & Value\\
\midrule 
Transmit power ${P_T}$ & $45$ dBm\\
Noise power $\sigma^2$ & $-90$ dBm\\
Carrier frequency $f_c$ & $28$ GHz\\
Termination iteration threshold $\delta$ & $3\times10^{-3}$\\
SNR threshold value $\gamma_{th}$ & $2$ dB\\
Number of path for BS-RIS channel $P$ & $4$\\
Number of path for RIS-UE channel $L$ & $4$\\
\bottomrule 
\end{tabular}}
\end{table}

\subsection{Simulation Setup }\label{S6-1}
\begin{figure}[t]
\centering
\includegraphics*[width=3.2in]{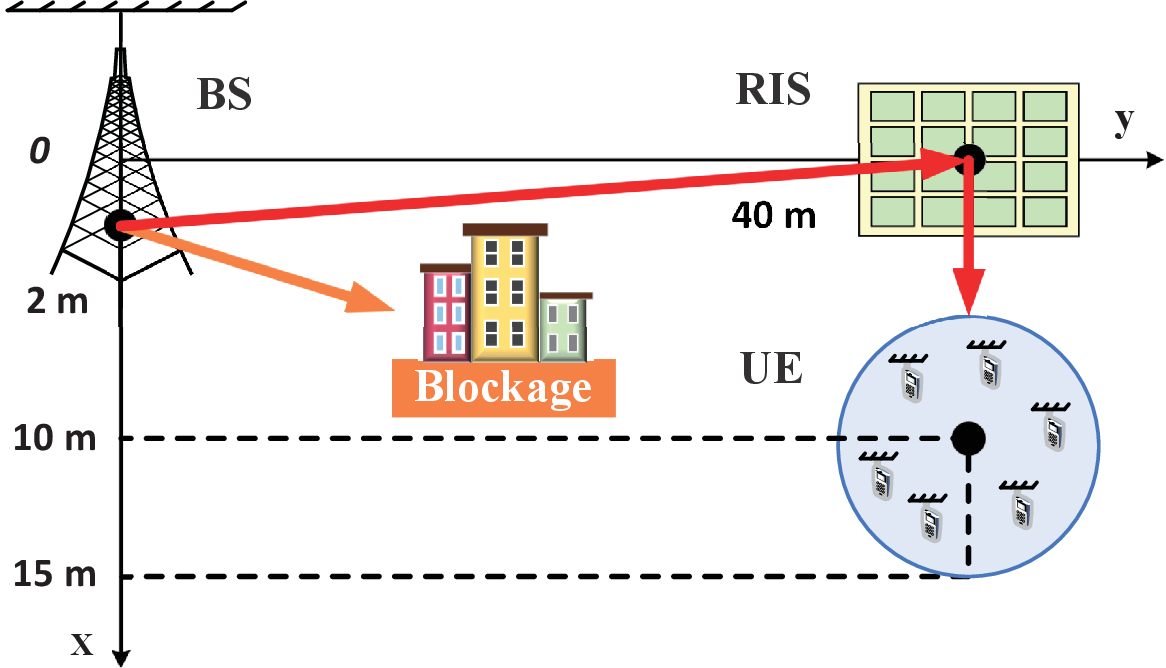}
\caption{Locations of communication nodes in the simulation.}
\label{fig:3}
\end{figure}
In the simulation, we establish a Cartesian coordinate system to describe the locations of communication nodes.
As shown in Fig.~\ref{fig:3}, the coordinates of the BS and the RIS
are given by (2 m, 0 m) and (0 m, 40 m), respectively. UEs are uniformly distributed in a circle centered at (10 m, 40 m) with a radius of 5 m. The height of the BS, the RIS, and the UEs is set to 10 m, 2.5 m, and 1.5 m, respectively. The BS-RIS channel and the RIS-UE channel are generated
according to the aforementioned SV model in LOS
scenarios, which can be further written as
\begin{equation}
\begin{split}
\mathbf{G}=\sqrt{\frac{N_tM}{P}}\bigg({\tilde \alpha_1}\mathbf{a}_{r}\left(M, \phi^r_{RIS,1},\zeta^r_{RIS,1}\right)\mathbf{a}_{t}^H\left(N_t, \psi^t_{BS,1}\right)\\
+\mathop \sum \limits_{i = 2}^P {\tilde \alpha_i}\mathbf{a}_{r}\left(M, \phi^r_{RIS,i},\zeta^r_{RIS,i}\right)\mathbf{a}_{t}^H\left(N_t, \psi^t_{BS,i}\right)\bigg),
\label{eq35}
\end{split}
\end{equation}
\begin{equation}
\begin{split}
\mathbf{H}_{r,t}=\sqrt{\frac{MN_r}{L}}\bigg({\tilde \beta_1}\mathbf{a}_{r}\left(N_r, \psi^r_{UE,1}\right)\mathbf{a}_{t}^H\left(M, \phi^t_{RIS,1},\zeta^t_{RIS,i}\right)\\
+\mathop \sum \limits_{i = 2}^L {\tilde \beta_i}\mathbf{a}_{r}\left(N_r, \psi^r_{UE,i}\right)\mathbf{a}_{t}^H\left(M, \phi^t_{RIS,i},\zeta^t_{RIS,i}\right)\bigg),
\label{eq36}
\end{split}
\end{equation}
where $\tilde \alpha_1$ ($\tilde \beta_1$) $\sim \mathcal{CN}(0,10^{-0.1\kappa})$ denotes the complex gain with the LOS component, $\tilde \alpha_i$ ($\tilde \beta_i$) $\sim \mathcal{CN}(0,10^{-0.1(\kappa+\mu)})$ denotes the complex gain with the $i$-th NLOS path, and $\kappa$ is the pathloss given by~\cite{38}
\begin{equation}
\kappa = a+10b\log_{10}(\tilde d)+\xi,
\label{eq37}
\end{equation}
in which $\tilde d$ is the distance between the transmitter and receiver, and $\xi \sim \mathcal{N}(0,\sigma^2_{\xi})$.
The values of $a$, $b$ and $\sigma_\xi$
are set as $a = 61.4$, $b = 2$, and $\sigma_\xi = 5.8 \text{dB}$ as suggested by
LOS real-world channel measurements~\cite{38}.
The Rician factor $\mu$ is set to 10,
which is defined as the ratio of the energy in the LOS path to the
sum of the energy in other NLOS paths~\cite{11,39}.
In the following simulations,
unless specified otherwise, we assume
$K = 6$, $M_a = M_b = 10$, $N_t = N_r = 64$, $b=3$, $T=100$, and $\mathcal{A}_{k,\text{max}}=\mathcal{A}_{\text{max}}=9, \forall k$. All simulation curves
are averaged over 100 independent channel realizations.
Other parameters are set as listed in Table~\ref{Simulation}.

To validate the system performance of the proposed algorithm, we compare it with the following baseline algorithms:
\begin{enumerate}
  \item \textbf{Random-RIS}: this algorithm randomly selects a feasible phase shift for each RIS element and keeps on using these phase shifts. Then, beamforming vectors are obtained by the hierarchical search method and the scheduling strategy is determined by Algorithm~\ref{alg.e}.

  \item \textbf{Random-BF}: this algorithm randomly chooses the codewords from the codebooks for beamforming at both the BS side and the scheduled UE side. The codebook consists of all the code words in the last layer of the hierarchical codebook. Then, the RIS reflection coefficients are adjusted by the local research method and the scheduling strategy is computed by Algorithm~\ref{alg.e}.

  \item \textbf{Round-Robin scheduling}: the only difference between this scheme and the proposed algorithm is the scheduling strategy. This scheme allocates time slots to UEs in descending order of data rates as in lines 1-7 of Algorithm~\ref{alg.e}, but it does not make further adjustments to the scheduling strategy.

\end{enumerate}

\subsection{Performance Evaluation }\label{S6-1}
\subsubsection{Impact of Maximum Tolerable AoI}
In Fig.~\ref{fig:5} and Fig.~\ref{fig:6}, we study
the impact of the maximum tolerable AoI $\mathcal{A}_{\text{max}}$ on the performance of the four schemes, which indicates the information freshness requirement of UEs.
Specifically, Fig.~\ref{fig:5} compares the sum rates of these schemes over $T$ time slots under different $\mathcal{A}_{\text{max}}$, and
Fig.~\ref{fig:6} compares the average system AoI of these schemes under different $\mathcal{A}_{\text{max}}$, which is defined as the average AoI of all UEs, (i.e., $\frac{1}{K}\sum_{k=1}^{K}\mathcal{A}_k$).
There are several important observations.
First, the Round-Robin scheduling algorithm achieves the lowest average AoI, but the sum rate and the average AoI do not change with $\mathcal{A}_{\text{max}}$.
The reason is that it cannot adjust the time slot allocation according to the information freshness constraint. In other words, the scheduling strategy is consistent under different information freshness requirements, which limits the sum rate.
In contrast, the proposed algorithm improves the sum rate performance at the cost of increasing the AoI while ensuring that the AoI constraints are satisfied. As $\mathcal{A}_{\text{max}}$ is increased, the algorithm can increase AoI accordingly to obtain a larger sum rate.
In addition, for the Random-BF scheme and Random-RIS scheme,
both the data rates and the average AoI beyond $\mathcal{A}_{\text{max}}$ are poor.
This is because the random beamforming or random RIS reflection coefficients severely degrades the received signal quality, and even makes most UEs unable to demodulate the transmit signal. In this case, we treat the data rates of these UEs as zero. Accordingly, these UEs cannot be scheduled and the AoI of these UEs keeps on accumulating over time, which results in poor sum rate and AoI performance.
\begin{figure}[t]
\centering
\includegraphics*[width=3.2in,height=2.5in]{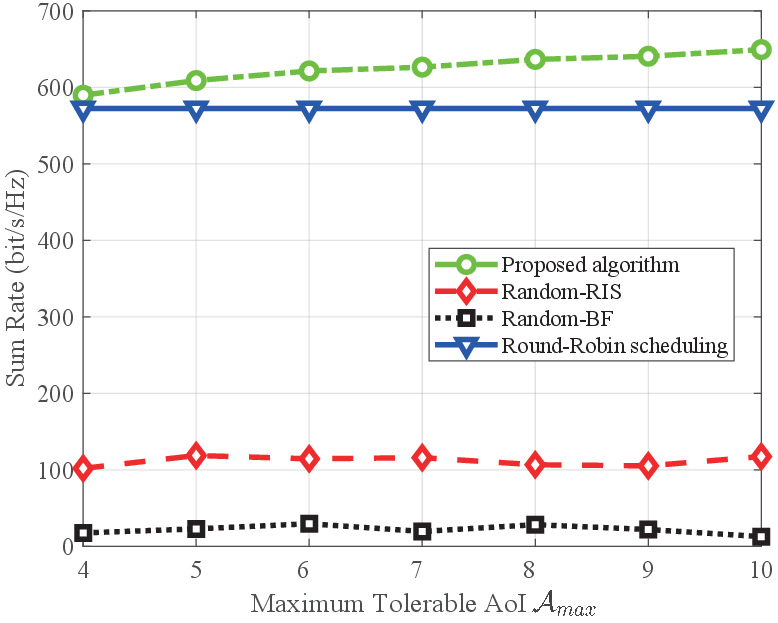}
\caption{Sum rate over $T$ time slots versus $\mathcal{A}_{\text{max}}$.}
\label{fig:5}
\end{figure}
\begin{figure}[t]
\centering
\includegraphics*[width=3.2in,height=2.5in]{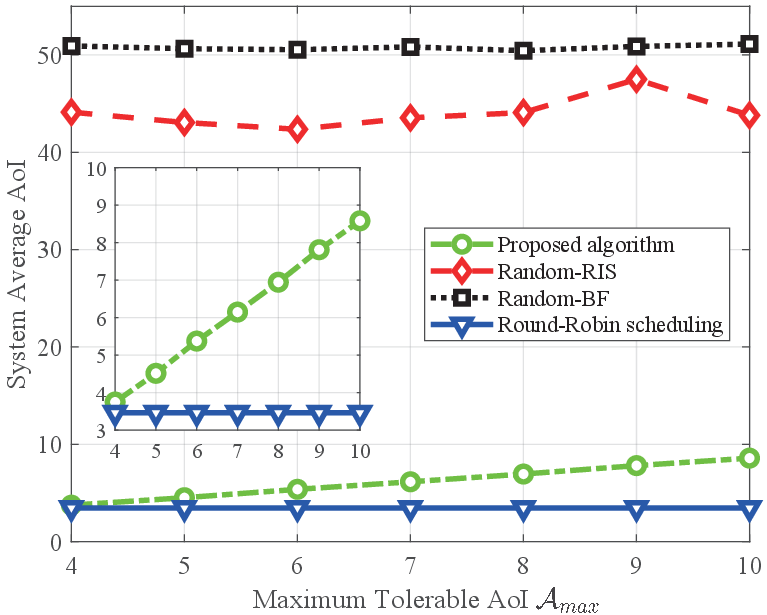}
\caption{System average AoI under different $\mathcal{A}_{\text{max}}$.}
\label{fig:6}
\end{figure}

In Fig.~\ref{fig:7}, we focus on the proposed algorithm and show the average AoI performance of each UE under different $\mathcal{A}_{\text{max}}$. Among all UEs, UE 2 has the highest data rate. First, we observe that the average AoI of each UE for different $\mathcal{A}_{\text{max}}$ does not exceed $\mathcal{A}_{\text{max}}$. Then, the larger the $\mathcal{A}_{\text{max}}$, the smaller the average AoI of UE 2 and the larger the average AoI of other UEs. This is because as $\mathcal{A}_{\text{max}}$ is increased, fewer time slots are needed to meet the information freshness requirements, so the proposed algorithm can allocate more time slots to UE 2 to enhance the sum rate over $T$ time slots, which reduces the AoI of UE 2. This further explains the increases of the sum rate and average AoI of the proposed algorithm in Fig.~\ref{fig:5} and Fig.~\ref{fig:6} with increased $\mathcal{A}_{\text{max}}$.

In Fig.~\ref{fig:8}, we consider the case where
there are two optional service types on UEs, which correspond to different information freshness requirements.
According to the service type of each UE, We divide UEs into two categories according to their information freshness requirements: UEs with high requirement and UEs with low requirement. The high-requirement corresponds to $\mathcal{A}_{k,\text{max}}=4$ and the low-requirement corresponds to $\mathcal{A}_{k,\text{max}}=9$. We plot the
sum rate over $T$ time slots while varying the number of UEs with high requirement from 0 to 6. When there are more high-requirement UEs in the system, the proposed algorithm needs to spend more time slots to satisfy the information freshness requirements, resulting in a lower sum rate. However, the proposed scheme still achieves the best performance among all the schemes.

\begin{figure}[!t]
\centering
\includegraphics*[width=3.2in,height=2.5in]{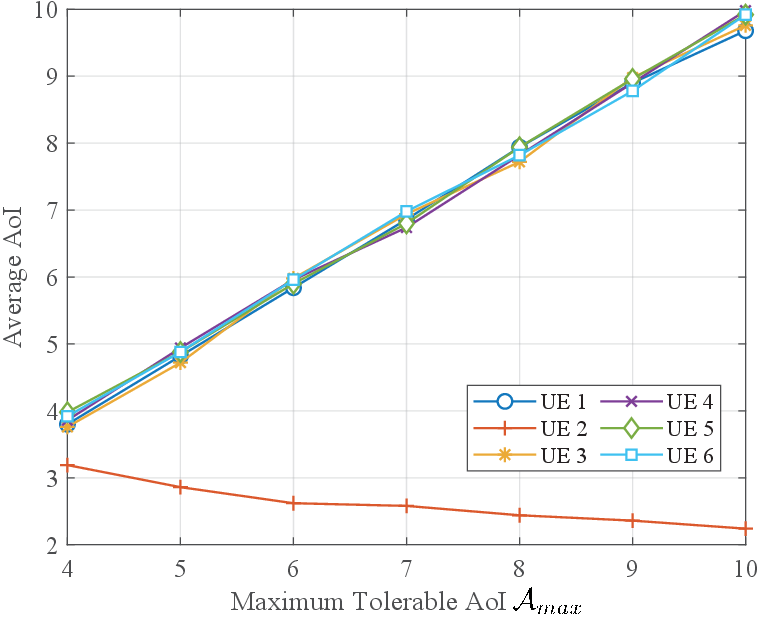}
\caption{Average AoI of each UE under different $\mathcal{A}_{\text{max}}$.}
\label{fig:7}
\end{figure}
\begin{figure}[!t]
\centering
\includegraphics*[width=3.2in,height=2.5in]{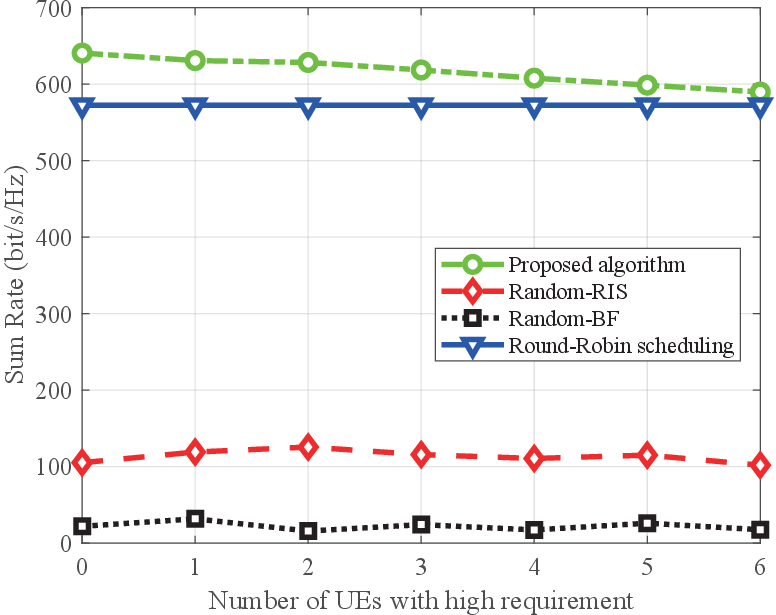}
\caption{Sum rate over $T$ time slots versus the number of UEs with high requirement.}
\label{fig:8}
\end{figure}

\subsubsection{Impact of Other Parameters}
In Fig.~\ref{fig:8b}, we examine the rate performance of each UE under the different number of iterations in Algorithm~\ref{alg.d}. We can see that the rate of all UEs converges to a stable value, which validates our convergence analysis in Section V-D. Besides, the number of iterations required for convergence is no more than 6, indicating that the BCD algorithm has a very fast convergence rate. Similar convergence rates can be seen in other related papers using the BCD algorithm, such as Fig. 3 and Fig. 4 in~\cite{40} and Fig. 3 in~\cite{41}. Such a fast convergence rate allows the algorithm to have reduced complexity.
%
%

In Fig.~\ref{fig:9}, we vary the number of UEs from 4 to 14 and compare the four schemes in terms of the sum rate over $T$ time slots.
Under the different number of UEs, we always set UE 2 as the UE with the highest rate.
It is observed that the proposed algorithm achieves the highest sum rate. As the number of UEs is increased, the sum rate of the proposed algorithm shows a decreasing trend. This is because the scheduling strategy needs to meet the information freshness requirements for more UEs within $T$ time slots, and the number of additional time slots allocated to the UE with the highest data rate is reduced accordingly.
In contrast, the Round-Robin scheduling scheme does not take into account the information freshness of UEs, so the changes in the sum rate are
only related to the rate performance of the added UEs.
In general, when the number of UEs is 14, the performance gap between the proposed algorithm and Round-Robin scheduling is $6.28\%$.
Further, a noticeable difference is observed between the proposed scheme and the other two schemes, i.e., Random-RIS and Random-BF, revealing the importance of jointly optimizing both RIS reflection coefficients and beamforming.
\begin{figure}[t]
\centering
\includegraphics*[width=3.2in,height=2.5in]{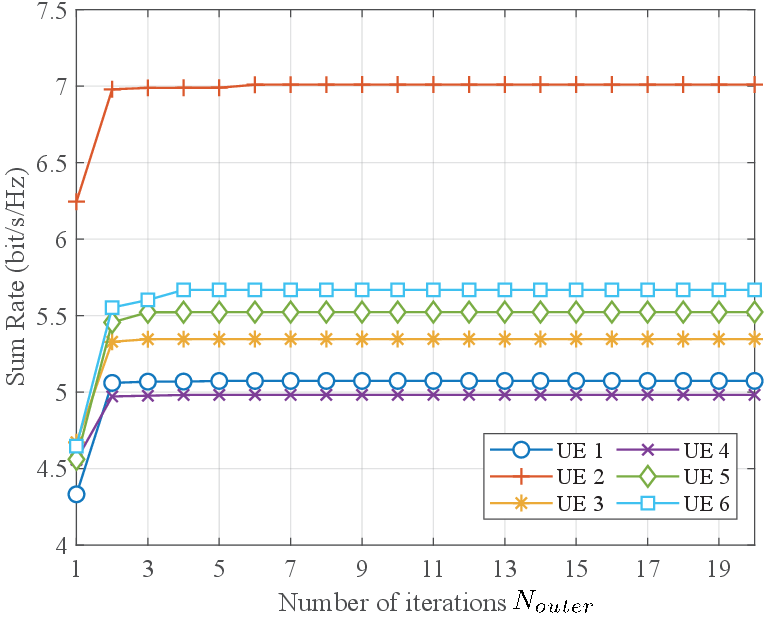}
\caption{Sum rate over $T$ time slots versus number of iteration $N_{outer}$.}
\label{fig:8b}
\end{figure}
\begin{figure}[!t]
\centering
\includegraphics*[width=3.2in,height=2.5in]{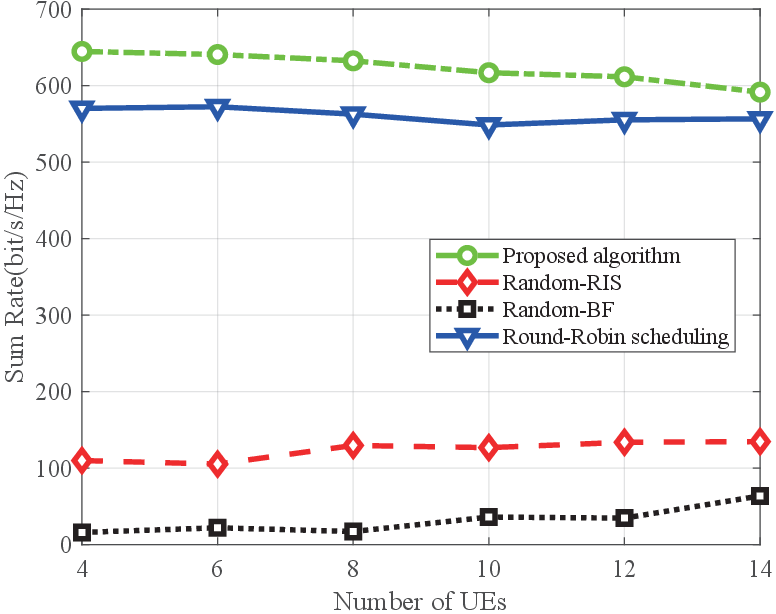}
\caption{Sum rate over $T$ time slots versus the number of UEs.}
\label{fig:9}
\end{figure}

In Fig.~\ref{fig:10}, we plot the sum rates of the four schemes over $T$ time slots while increasing the bit-quantization number from 1 to 6.
As seen from the given results, the proposed algorithm outperforms the baseline schemes.
The sum rates of the proposed scheme and the Round-Robin scheduling scheme gradually increase as $b$ grows from 1 to 3, and then basically remain unchanged from 3 to 6. This shows that the system performance tends to be saturated when the number of quantization bits exceeds 3.
The performance of the Random-RIS scheme is similar to that of the proposed scheme when $b$ is 1.
However, with the increase of $b$, it is more difficult to obtain an effective reflection coefficients matrix by Random-RIS.
So the gap with the proposed algorithm widens when $b>1$, and the sum rate fluctuates around a lower value.
In addition, for Random-BF, when $b=1$, the sum rate is close to 0, which means there are few UEs in the system which can reliably demodulate the transmit signal. As $b$ increases, RIS can provide performance gain for reliable demodulation.
However, due to the random beamforming, the beams between BS, RIS, and the scheduled UE are not well aligned, which impedes the growth of the sum rate.
\begin{figure}[!t]
\centering
\includegraphics*[width=3.2in,height=2.5in]{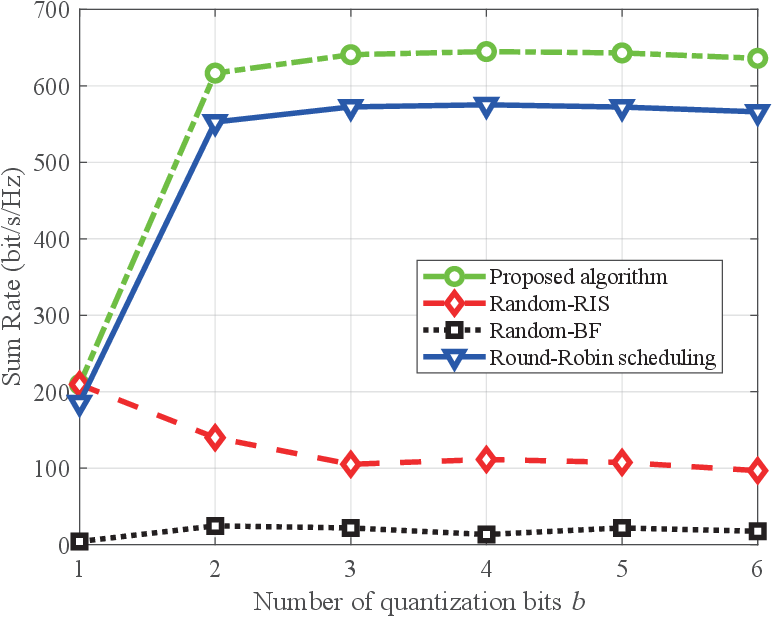}
\caption{Sum rate over $T$ time slots versus the number of quantization bits $b$.}
\label{fig:10}
\end{figure}

In Fig.~\ref{fig:11}, we plot the sum rates over $T$ time slots of the four schemes versus the number of RIS elements $M$. We see that the proposed algorithm outperforms the others as $M$ is increased from 36 to 256. Then, all these four schemes show an increasing trend with $M$, which indicates that we can enhance the system sum rate by deploying RIS with more elements.
Note that the increase of the Random-RIS scheme is due to the aperture gain of the RIS. The larger the RIS aperture, the more signal power in the BS-RIS link can be collected by the RIS.
Further, the gaps between the proposed algorithm and the other two schemes, i.e., Random-RIS and Random-BF, gradually widen as $M$ is increased.
Therefore, we need to design the joint RIS and beamforming optimization more carefully when more RIS elements are available.
\begin{figure}[!t]
\centering
\includegraphics*[width=3.2in,height=2.5in]{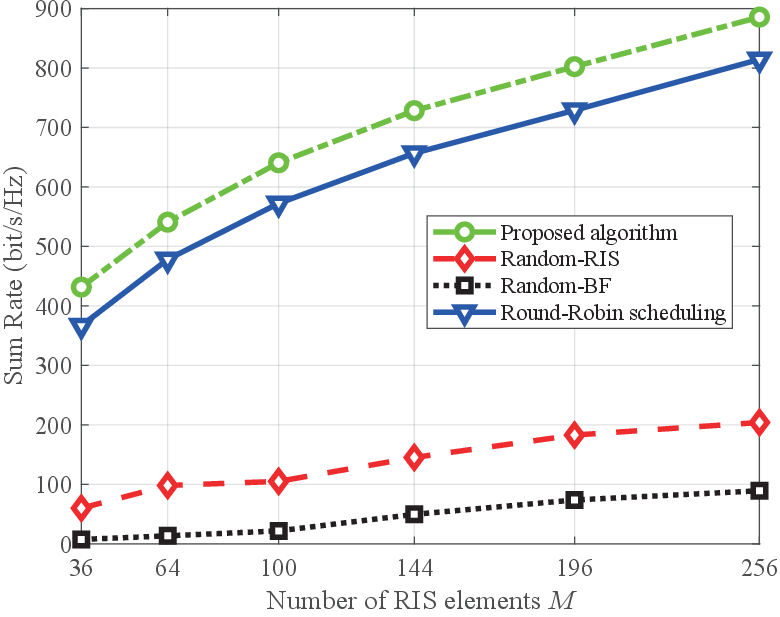}
\caption{Sum rate over $T$ time slots versus the number of RIS elements $M$.}
\label{fig:11}
\end{figure}

In Fig.~\ref{fig:12}, we plot the sum rate over $T$ time slots versus the number of transmit antenna $N_t$ at the BS, which is varied from 16 to 512. It can be seen that the sum rate increases with the
number of transmit antennas for the proposed algorithm, the Round-Robin scheduling scheme, and the Random-RIS scheme.
The growths slow down with further increased number of antennas. When $N_t$ is less than 128, the increase of $N_t$ results in the most significant improvement in the sum rate.
However, since the Random-BF scheme cannot provide a stable beamforming gain for the system, the increase in the number of transmit antennas has little impact on its performance.
From this figure, we observe that the proposed algorithm has the highest sum rate than the other schemes. In general, when the number of transmit antennas is 512, the performance gap between the proposed algorithm and the three baseline schemes is $9.5\%$, $231.7\%$, and $5039.3\%$, respectively.
\begin{figure}[!t]
\centering
\includegraphics*[width=3.2in,height=2.5in]{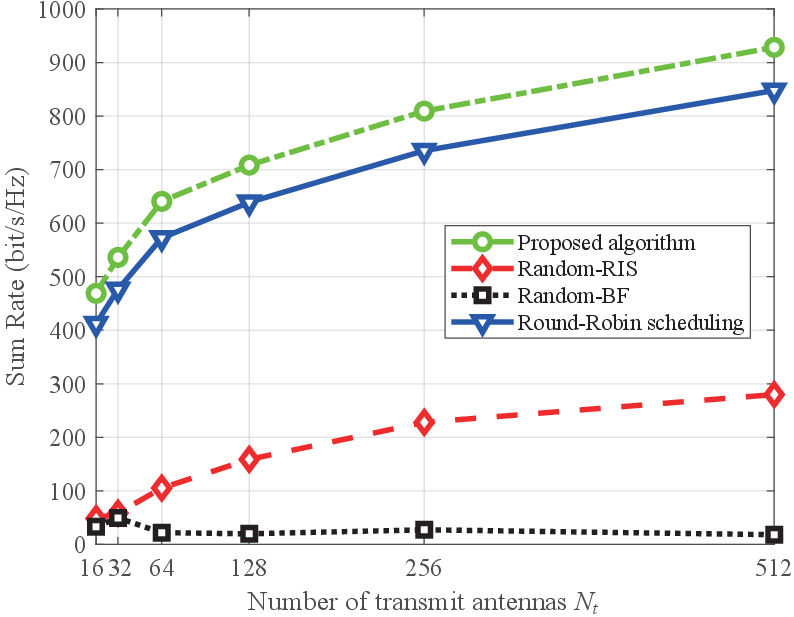}
\caption{Sum rate over $T$ time slots versus the number of transmit antennas $N_t$.}
\label{fig:12}
\end{figure}

In Fig.~\ref{fig:13}, we vary $T$ from 50 to 300 and compare the four schemes in terms of the sum rate over $T$ time slots. We can see that except for the Random-BF scheme, where the beams cannot be well aligned, the sum rates of all the other schemes increase linearly with $T$.
The average sum rate of $T$ time slots for the three schemes, i.e., $\frac{1}{T} \sum_{t=1}^{T} R_t$, can be calculated as
5.97 bit/s/Hz, 5.69 bit/s/Hz, and 1.07 bit/s/Hz, respectively.
Apparently, the proposed algorithm has the best sum rate performance. The reason is that the proposed scheme can schedule the UE with the highest data rate as many times as possible compared to Round-Robin scheduling, and it can achieve the joint optimization of beamforming and RIS reflection coefficients compared to Random-BF and Random-RIS.
\begin{figure}[!t]
\centering
\includegraphics*[width=3.2in,height=2.5in]{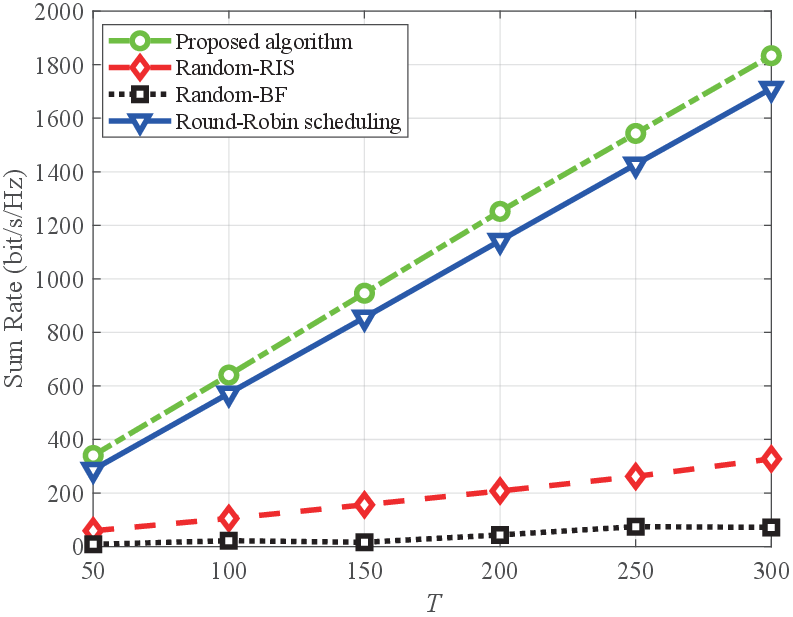}
\caption{Sum rate over $T$ time slots versus $T$.}
\label{fig:13}
\end{figure}

\section{Conclusions}\label{S7}

In this paper, we investigated the sum rate maximization problem in RIS-assisted mmWave MIMO communication systems, where the information freshness requirements of all UEs should be satisfied. To solve this problem, we adopted the BCD method to jointly optimize RIS reflection coefficients and beamforming, and the heuristic scheduling algorithm to design the scheduling strategy. In particular, considering the difficulty of channel estimation in such systems, we utilized the hierarchical search method to update beamforming and the local search method to update RIS reflection coefficients. Simulation results showed that our algorithm can not only ensure the information freshness of UEs but also have the best sum rate performance. In future work, we will consider the case of scheduling multiple UEs in each time slot, where we will jointly design beamforming vectors, RIS phase shifts, and scheduling strategies to combat inter-user interference and satisfy the requirements of information freshness. In addition, we will extend this work to multi-cell multi-RIS scenarios in the future. With the joint design and optimization for multiple BSs and RISs, the information freshness requirement of UEs can be more effectively satisfied, and the system sum rate can be further improved.

\bibliographystyle{IEEEtranTCOM}

\bibliographystyle{IEEEtran}

\begin{IEEEbiography}[{\includegraphics[width=1in,trim=0 15 0 0,clip,keepaspectratio]{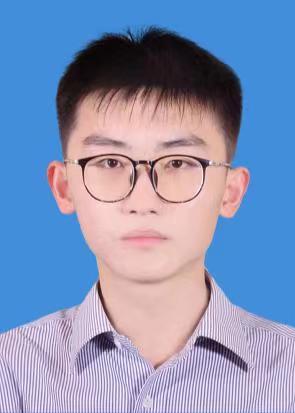}}] {Ziqi Guo} was born in Shandong, China, in 2000. He received the B.E. degree in communication engineering from Beijing Jiaotong University, Beijing, China, in 2021. He is currently working toward the M.S. degree with the State Key Laboratory of Advanced Rail Autonomous Operation, Beijing Jiaotong University, Beijing, China. His research interests include mmWave wireless communications and reconfigurable intelligent surface.
\end{IEEEbiography}

\begin{IEEEbiography}[{\includegraphics[width=1in,trim=0 15 0 0,clip,keepaspectratio]{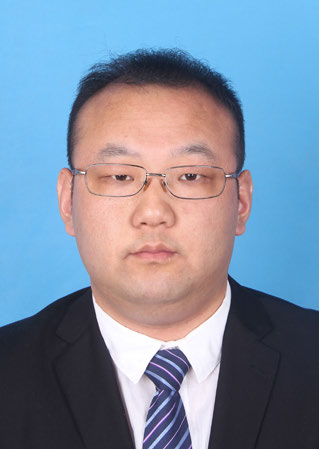}}]{Yong Niu}(Senior Member, IEEE) received the B.E. degree in electrical engineering from Beijing Jiaotong University, China, in 2011, and the Ph.D. degree in electronic engineering from Tsinghua University, Beijing, China, in 2016.

From 2014 to 2015, he was a Visiting Scholar with the University of Florida, Gainesville, FL, USA. He is currently an Associate Professor with State Key Laboratory of Advanced Rail Autonomous Operation, Beijing Jiaotong University. His research interests include networking and communications, including millimeter wave communications, device-to-device communication, mediumaccess control, and software-defined networks. He has served as a Technical Program Committee Member for IWCMC 2017, VTC 2018-Spring, IWCMC 2018, INFOCOM 2018, and ICC 2018. He was the Session Chair for IWCMC 2017. He was a recipient of the Ph.D. National Scholarship of China in 2015, the Outstanding Ph.D. Graduates and Outstanding Doctoral Thesis of Tsinghua University in 2016, the Outstanding Ph.D. Graduates of Beijing in 2016, and the Outstanding Doctorate Dissertation Award from the Chinese Institute of Electronics in 2017, and the 2018 International Union of Radio Science Young Scientist Award.
\end{IEEEbiography}

\begin{IEEEbiography}[{\includegraphics[width=1in,trim=0 15 0 0,clip,keepaspectratio]{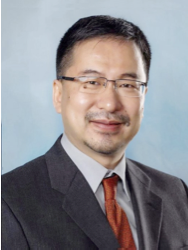}}]{Shiwen Mao}(Fellow, IEEE) received his Ph.D. in electrical engineering from Polytechnic University, Brooklyn, NY in 2004. Currently, he is a Professor and Earle C. Williams Eminent Scholar Chair in Electrical and Computer Engineering at Auburn University, Auburn, AL. His research interests include wireless networks, multimedia communications, and smart grid. He is a Distinguished Lecturer of IEEE Communications Society (2021-2022) and IEEE Council of RFID (2021-2022), and a Distinguished Lecturer (2014-2018) and a Distinguished Speaker of IEEE Vehicular Technology Society (2018-2021). He is on the Editorial Board of IEEE/CIC China Communications, IEEE Transactions on Wireless Communications, IEEE Internet of Things Journal, IEEE Open Journal of the Communications Society, ACM GetMobile, IEEE Transactions on Cognitive Communications and Networking, IEEE Transactions on Network Science and Engineering, IEEE Transactions on Mobile Computing, IEEE Multimedia, IEEE Network, and IEEE Networking Letters. He is a co-recipient of the 2021 IEEE Internet of Things Journal Best Paper Award, the 2021 IEEE Communications Society Outstanding Paper Award, the IEEE Vehicular Technology Society 2020 Jack Neubauer Memorial Award, the IEEE ComSoc MMTC 2018 Best Journal Paper Award and the 2017 Best Conference Paper Award, the Best Demo Award of IEEE SECON 2017, the Best Paper Awards of IEEE GLOBECOM 2019, 2016, and 2015, IEEE WCNC 2015, and IEEE ICC 2013, and the 2004 IEEE Communications Society Leonard G. Abraham Prize in the Field of Communications Systems. He is a Fellow of the IEEE and a Member of the ACM.
\end{IEEEbiography}

\begin{IEEEbiography}[{\includegraphics[width=1in,trim=0 15 0 0,clip,keepaspectratio]{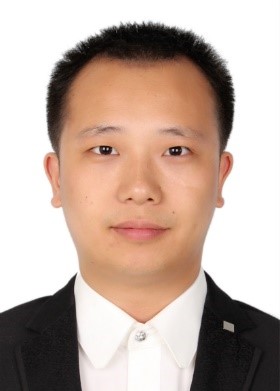}}]{Changming Zhang} received the B.S. degree from the Department of Electronic Information Science and Technology, Beijing Normal University, Beijing, China, in 2010, and the Ph.D. degree from the Department of Electronic Engineering, Tsinghua University, Beijing, in 2015. He is currently a research expert with the Research Institute of Intelligent Networks, Zhejiang Lab, Hangzhou, China. His research interests include millimeter-wave and terahertz wireless communications, including huge-capacity transmission, complex digital signal processing, and broadband wireless networks.
\end{IEEEbiography}

\begin{IEEEbiography}[{\includegraphics[width=1in,trim=0 15 0 0,clip,keepaspectratio]{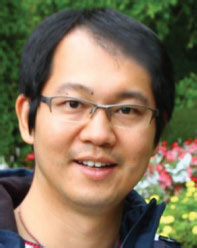}}]{Ning Wang}(Member, IEEE) received the B.E. degree in communication engineering from Tianjin University, Tianjin, China, in 2004, the M.A.Sc. degree in electrical engineering from The University of British Columbia, Vancouver, BC, Canada,
in 2010, and the Ph.D. degree in electrical engineering from the University of Victoria, Victoria, BC, Canada, in 2013. From 2004 to 2008, he was with the China Information Technology Design and Consulting Institute, as a Mobile Communication System Engineer, specializing in planning and design of commercial mobile communication networks, network traffic analysis, and radio network optimization. From 2013 to 2015, he was a Post-Doctoral Research Fellow with the Department of Electrical and Computer Engineering, The University of British Columbia. Since 2015, he has been with the School of Information Engineering, Zhengzhou University, Zhengzhou, China, where he is currently an Associate Professor. He also holds adjunct appointments with the Department of Electrical and Computer Engineering, McMaster University, Hamilton, ON, Canada, and the Department of Electrical and Computer Engineering, University of Victoria, Victoria, BC, Canada. His research interests include resource allocation and security designs of future cellular networks, channel modeling for wireless communications, statistical signal processing, and cooperative wireless communications. He has served on the technical program committees of international conferences, including the IEEE GLOBECOM, IEEE ICC, IEEE WCNC, and CyberC. He was on the Finalist of the Governor Generals Gold Medal for Outstanding Graduating Doctoral Student from the University of Victoria in 2013.
\end{IEEEbiography}

\begin{IEEEbiography}[{\includegraphics[width=1in,trim=0 15 0 0,clip,keepaspectratio]{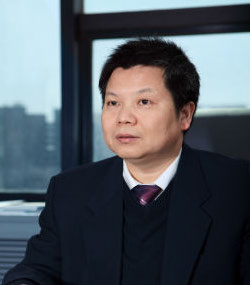}}]{Zhangdui Zhong}(Fellow, IEEE) received the B.E. and M.S. degrees from Beijing Jiaotong University, Beijing, China, in 1983 and 1988, respectively.

He is currently a Professor and an Advisor of Ph.D. students with Beijing Jiaotong University, where he is also the Chief Scientist of State Key Laboratory of Advanced Rail Autonomous Operation. He is the Director of the Innovative Research Team, Ministry of Education, Beijing, and the Chief Scientist of the Ministry of Railways, Beijing. He is an Executive Council Member of the Radio Association of China, Beijing, and the Deputy Director of the Radio Association, Beijing. His research interests include wireless communications for railways, control theory, and techniques for railways, and GSM-R systems. His research has been widely used in railway engineering, such as the Qinghai-Xizang railway, Datong-Qinhuangdao Heavy Haul railway, and many high-speed railway lines in China. He has authored or coauthored seven books, five invention patents, and over 200 scientific research papers in his research area. He was a recipient of the Mao YiSheng Scientific Award of China, Zhan TianYou Railway Honorary Award of China, and Top 10 Science/Technology Achievements Award of Chinese Universities.
\end{IEEEbiography}

\begin{IEEEbiography}[{\includegraphics[width=1in,trim=0 15 0 0,clip,keepaspectratio]{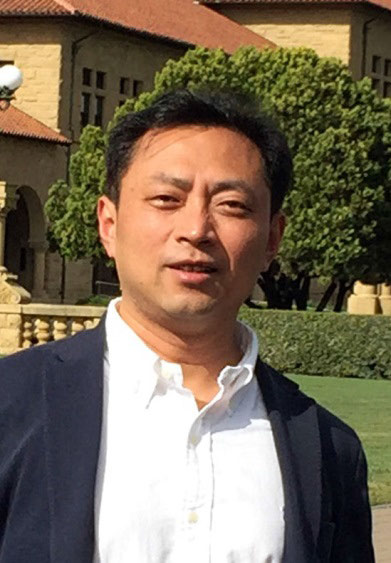}}]{Bo Ai}(Fellow, IEEE) received the M.S. and Ph.D. degrees from Xidian University, Xi'an, China, in 2002 and 2004, respectively.

He was an Excellent Post-Doctoral Research Fellow at Tsinghua University, Beijing, China, in 2007. He was a Visiting Professor at the EE Department, Stanford University, in 2015. He is currently working at Beijing Jiaotong University as a Full Professor and a Ph.D. Candidate Advisor. He is also the Deputy Director of State Key Laboratory of Advanced Rail Autonomous Operation and the Deputy Director of the International Joint Research Center. He is one of the main responsible people for Beijing "Urban rail operation control system" International Science and Technology Cooperation Base and the Member of the Innovative Engineering based jointly granted by Chinese Ministry of Education and the State Administration of Foreign Experts Affairs. He has authored/coauthored eight books and published over 300 academic research papers in his research area. He has hold 26 invention patents. He has been the research team leader for 26 national projects and has won some important scientific research prizes. He has been notified by Council of Canadian Academies (CCA), that based on Scopus database. He has been listed as one of the Top 1 authors in his field all over the world. He has also been Feature Interviewed by ELECTRONICS LETTERS (IET). His interests include the research and applications of channel measurement and channel modeling, dedicated mobile communications for rail traffic systems. He has received some important scientific research prizes.
\end{IEEEbiography}

\bibliographystyle{IEEEtran}

\begin{thebibliography}{10}
\small

\bibitem{01}
ITU-R M.2370-0, ``IMT traffic estimates for the years 2020 to 2030,'' July 2015.

\bibitem{02}
S. Rangan, T. S. Rappaport and E. Erkip, ``Millimeter-Wave Cellular Wireless Networks: Potentials and Challenges,'' \emph{Proceedings of the IEEE}, vol.~102, no.~3, pp. 366--385, March 2014.

\bibitem{02b}
A. Ghosh et al., ``Millimeter-Wave Enhanced Local Area Systems: A High-Data-Rate Approach for Future Wireless Networks,'' \emph{IEEE Journal on Selected Areas in Communications}, vol.~32, no.~6, pp. 1152--1163, June 2014.

\bibitem{02c}
S. Sun, T. S. Rappaport, M. Shafi, P. Tang, J. Zhang and P. J. Smith, ``Propagation Models and Performance Evaluation for 5G Millimeter-Wave Bands,'' \emph{IEEE Transactions on Vehicular Technology}, vol.~67, no.~9, pp. 8422--8439, Sept. 2018.


\bibitem{08}
R. D. Yates, Y. Sun, D. R. Brown, S. K. Kaul, E. Modiano and S. Ulukus, ``Age of Information: An Introduction and Survey,'' \emph{IEEE Journal on Selected Areas in Communications}, vol.~39, no.~5, pp. 1183--1210, May 2021.

\bibitem{08b}
Y. Sun, E. Uysal-Biyikoglu, R. D. Yates, C. E. Koksal and N. B. Shroff, ``Update or Wait: How to Keep Your Data Fresh,'' \emph{IEEE Transactions on Information Theory}, vol.~63, no.~11, pp. 7492--7508, Nov. 2017.



\bibitem{03}
A. Alkhateeb, J. Mo, N. Gonzalez-Prelcic and R. W. Heat, ``MIMO Precoding and Combining Solutions for Millimeter-Wave Systems,'' \emph{IEEE Communications Magazine}, vol.~52, no.~12, pp. 122-131, Dec. 2014.



\bibitem{04}
Y. Niu, Y. Li, D. Jin, L. Su, and A.V. Vasilakos, ``A survey of millimeter wave communications (mmWave) for 5G: opportunities and challenges,'' \emph{Springer Wireless networks}, vol.~21, no.~8, pp. 2657--2676, Apr. 2015.


\bibitem{05}
P. Wang, J. Fang, X. Yuan, Z. Chen and H. Li, ``Intelligent Reflecting Surface-Assisted Millimeter Wave Communications: Joint Active and Passive Precoding Design,'' \emph{IEEE Transactions on Vehicular Technology}, vol.~69, no.~12, pp. 14960--14973, Dec. 2020.


\bibitem{06}
Q. Wu, S. Zhang, B. Zheng, C. You and R. Zhang, ``Intelligent Reflecting Surface-Aided Wireless Communications: A Tutorial,'' \emph{IEEE Transactions on Communications}, vol.~69, no.~5, pp. 3313--3351, May 2021.




\bibitem{09}
P. Wang, J. Fang, W. Zhang, Z. Chen, H. Li and W. Zhang, ``Beam Training and Alignment for RIS-Assisted Millimeter-Wave Systems: State of the Art and Beyond,'' \emph{IEEE Wireless Communications}, vol.~29, no.~6, pp. 64--71, Dec. 2022.



\bibitem{10}

N. S. Perovic, M. D. Renzo and M. F. Flanagan, ``Channel Capacity Optimization Using Reconfigurable Intelligent Surfaces in Indoor mmWave Environments,'' in \emph{Proc. IEEE ICC}, Dublin, Ireland, June 2020, pp. 1--7.

\bibitem{11}

P. Wang, J. Fang, L. Dai and H. Li, ``Joint Transceiver and Large Intelligent Surface Design for Massive MIMO mmWave Systems,'' \emph{IEEE Transactions on Wireless Communications}, vol.~20, no.~2, pp. 1052--1064, Feb. 2021.


\bibitem{12}
C. Feng, W. Shen, J. An and L. Hanzo, ``Joint Hybrid and Passive RIS-Assisted Beamforming for mmWave MIMO Systems Relying on Dynamically Configured Subarrays,'' \emph{IEEE Internet of Things Journal}, vol.~9, no.~15, pp. 13913--13926, Aug. 2022.


\bibitem{13}

R. Li, B. Guo, M. Tao, Y. -F. Liu and W. Yu, ``Joint Design of Hybrid Beamforming and Reflection Coefficients in RIS-Aided mmWave MIMO Systems,'' \emph{IEEE Transactions on Communications}, vol.~70, no.~4, pp. 2404--2416, April 2022.




\bibitem{14}
P. Wang, J. Fang, W. Zhang and H. Li, ``Fast Beam Training and Alignment for IRS-Assisted Millimeter Wave/Terahertz Systems,'' \emph{IEEE Transactions on Wireless Communications}, vol.~21, no.~4, pp. 2710--2724, April 2022.


\bibitem{15}
X. Wei, L. Dai, Y. Zhao, G. Yu and X. Duan, ``Codebook design and beam training for extremely large-scale RIS: Far-field or near-field?,'' \emph{China Communications}, vol.~19, no.~6, pp. 193--204, June 2022.





\bibitem{16}

W. Wang and W. Zhang, ``Joint Beam Training and Positioning for Intelligent Reflecting Surfaces Assisted Millimeter Wave Communications,'' \emph{IEEE Transactions on Wireless Communications}, vol.~20, no.~10, pp. 6282--6297, Oct. 2021.


\bibitem{17}

Q. He, D. Yuan and A. Ephremides, ``Optimal Link Scheduling for Age Minimization in Wireless Systems,'' \emph{IEEE Transactions on Information Theory}, vol.~64, no.~7, pp. 5381--5394, July 2018.

\bibitem{18}

I. Kadota, A. Sinha and E. Modiano, ``Scheduling Algorithms for Optimizing Age of Information in Wireless Networks With Throughput Constraints,'' \emph{IEEE/ACM Transactions on Networking}, vol.~27, no.~4, pp. 1359--1372, Aug. 2019.


\bibitem{19}

Q. Liu, H. Zeng and M. Chen, ``Minimizing AoI With Throughput Requirements in Multi-Path Network Communication,'' \emph{IEEE/ACM Transactions on Networking}, vol.~30, no.~3, pp. 1203--1216, June 2022.


\bibitem{20}

R. V. Bhat, R. Vaze and M. Motani, ``Throughput Maximization With an Average Age of Information Constraint in Fading Channels,'' \emph{IEEE Transactions on Wireless Communications}, vol.~20, no.~1, pp. 481--494, Jan. 2021.

\bibitem{21}

F. Wu, H. Zhang, J. Wu, Z. Han, H. V. Poor and L. Song, ``UAV-to-Device Underlay Communications: Age of Information Minimization by Multi-Agent Deep Reinforcement Learning,'' \emph{IEEE Transactions on Communications}, vol.~69, no.~7, pp. 4461--4475, July 2021.


\bibitem{22}

T. D. P. Perera, D. N. K. Jayakody, I. Pitas and S. Garg, ``Age of Information in SWIPT-Enabled Wireless Communication System for 5GB,'' \emph{IEEE Wireless Communications}, vol.~27, no.~5, pp. 162--167, Oct. 2020.


\bibitem{23}
A. Muhammad, I. Sorkhoh, M. Samir, D. Ebrahimi and C. Assi, ``Minimizing Age of Information in Multiaccess-Edge-Computing-Assisted IoT Networks,'' \emph{IEEE Internet of Things Journal}, vol.~9, no.~15, pp. 13052--13066, Aug. 2022.

\bibitem{24}

J. Lee, D. Niyato, Y. L. Guan and D. I. Kim, ``Learning to Schedule Joint Radar-Communication With Deep Multi-Agent Reinforcement Learning,'' \emph{IEEE Transactions on Vehicular Technology}, vol.~71, no.~1, pp. 406--422, Jan. 2022.

\bibitem{25}

I. Sorkhoh, M. A. Arfaoui, M. Khabbaz and C. Assi, ``Optimizing Information Freshness in RIS-Assisted Cooperative Autonomous Driving,'' in \emph{Proc. IEEE ICC}, Seoul, Korea, Republic of, May 2022, pp. 1518--1523.



\bibitem{26}


A. Muhammad, M. Elhattab, M. Shokry and C. Assi, ``Leveraging Reconfigurable Intelligent Surface to Minimize Age of Information in Wireless Networks,'' in \emph{Proc. IEEE ICC}, Seoul, Korea, Republic of, May 2022, pp. 2525--2530.


\bibitem{27}

M. Samir, M. Elhattab, C. Assi, S. Sharafeddine and A. Ghrayeb, ``Optimizing Age of Information Through Aerial Reconfigurable Intelligent Surfaces: A Deep Reinforcement Learning Approach,'' \emph{IEEE Transactions on Vehicular Technology}, vol.~70, no.~4, pp. 3978--3983, April 2021.


\bibitem{28}

X. Fan, M. Liu, Y. Chen, S. Sun, Z. Li and X. Guo, ``RIS-Assisted UAV for Fresh Data Collection in 3D Urban Environments: A Deep Reinforcement Learning Approach,'' \emph{IEEE Transactions on Vehicular Technology}, vol.~72, no.~1, pp. 632--647, Jan. 2023.

\bibitem{29}

X. Feng, S. Fu, F. Fang and F. R. Yu, ``Optimizing Age of Information in RIS-Assisted NOMA Networks: A Deep Reinforcement Learning Approach,'' \emph{IEEE Wireless Communications Letters}, vol.~11, no.~10, pp. 2100--2104, Oct. 2022.

\bibitem{30}

W. Lyu, Y. Xiu, J. Zhao and Z. Zhang, ``Optimizing the Age of Information in RIS-Aided SWIPT Networks,'' \emph{IEEE Transactions on Vehicular Technology}, vol.~72, no.~2, pp. 2615--2619, Feb. 2023.

\bibitem{31}

Z. Shi, H. Wang, Y. Fu, X. Ye, G. Yang and S. Ma, ``Outage Performance and AoI Minimization of HARQ-IR-RIS Aided IoT Networks,'' \emph{IEEE Transactions on Communications}, vol.~71, no.~3, pp. 1740--1754, March 2023.

\bibitem{31b}

Q. Wu and R. Zhang, ``Beamforming Optimization for Wireless Network Aided by Intelligent Reflecting Surface With Discrete Phase Shifts,'' \emph{IEEE Transactions on Communications}, vol.~68, no.~3, pp. 1838--1851, March 2020.

\bibitem{32}

Q. Wu and R. Zhang, ``Intelligent Reflecting Surface Enhanced Wireless Network via Joint Active and Passive Beamforming,'' \emph{IEEE Transactions on Wireless Communications}, vol.~18, no.~11, pp. 5394--5409, Nov. 2019.


\bibitem{33}

Q. Wu and R. Zhang, ``Towards Smart and Reconfigurable Environment: Intelligent Reflecting Surface Aided Wireless Network,'' \emph{IEEE Communications Magazine}, vol.~58, no.~1, pp. 106--112, Jan. 2020.

\bibitem{33b}

M. Gao, B. Ai, Y. Niu, Z. Han and Z. Zhong, ``IRS-Assisted High-Speed Train Communications: Outage Probability Minimization with Statistical CSI,'' in \emph{Proc. IEEE ICC}, Montreal, QC, Canada, June 2021, pp. 1--6.

\bibitem{07}
Y. Chen et al., ``Reconfigurable Intelligent Surface Assisted Device-to-Device Communications,'' \emph{IEEE Transactions on Wireless Communications}, vol.~20, no.~5, pp. 2792--2804, May 2021.




\bibitem{35}

O. E. Ayach, S. Rajagopal, S. Abu-Surra, Z. Pi and R. W. Heath, ``Spatially Sparse Precoding in Millimeter Wave MIMO Systems,'' \emph{IEEE Transactions on Wireless Communications}, vol.~13, no.~3, pp. 1499--1513, March 2014.

\bibitem{34}

N. Huang, T. Wang, Y. Wu, Q. Wu and T. Q. S. Quek, ``Integrated Sensing and Communication Assisted Mobile Edge Computing: An Energy-Efficient Design via Intelligent Reflecting Surface,'' \emph{IEEE Wireless Communications Letters}, vol.~11, no.~10, pp. 2085--2089, Oct. 2022.

\bibitem{36}

Z. Xiao, T. He, P. Xia and X. -G. Xia, ``Hierarchical Codebook Design for Beamforming Training in Millimeter-Wave Communication,'' \emph{IEEE Transactions on Wireless Communications}, vol.~15, no.~5, pp. 3380--3392, May 2016.


\bibitem{37}

C. Liu, M. Li, S. V. Hanly, I. B. Collings and P. Whiting, ``Millimeter Wave Beam Alignment: Large Deviations Analysis and Design Insights,'' \emph{IEEE Journal on Selected Areas in Communications}, vol.~35, no.~7, pp. 1619--1631, July 2017.


\bibitem{38}

M. R. Akdeniz et al., ``Millimeter Wave Channel Modeling and Cellular Capacity Evaluation,'' \emph{IEEE Journal on Selected Areas in Communications}, vol.~32, no.~6, pp. 1164--1179, June 2014.





\bibitem{39}

M. K. Samimi, G. R. MacCartney, S. Sun and T. S. Rappaport, ``28 GHz Millimeter-Wave Ultrawideband Small-Scale Fading Models in Wireless Channels,'' in \emph{Proc. IEEE VTC Spring}, Nanjing, China, May 2016, pp. 1--6.


\bibitem{40}
R. Li, B. Guo, M. Tao, Y. -F. Liu and W. Yu, ``Joint Design of Hybrid Beamforming and Reflection Coefficients in RIS-Aided mmWave MIMO Systems,'' \emph{IEEE Transactions on Communications}, vol.~70, no.~4, pp. 2404--2416, April 2022.

\bibitem{41}
H. Gao, K. Cui, C. Huang and C. Yuen, ``Robust Beamforming for RIS-Assisted Wireless Communications With Discrete Phase Shifts,'' \emph{IEEE Wireless Communications Letters}, vol.~10, no.~12, pp. 2619--2623, Dec. 2021.



\end{thebibliography}

\end{document}